\documentclass{ecai}

\usepackage{amsmath, amssymb, amsthm}
\usepackage{paralist}
\usepackage{xspace}
\usepackage{tikz}
\usetikzlibrary{calc,shapes,arrows}
\usepackage{thmtools, thm-restate}
\usepackage{stmaryrd}
\usepackage{algorithm}
\usepackage[noend]{algpseudocode}

\usepackage{caption}
\usepackage{todonotes}

\newcommand{\mc}[1]{\mathcal{#1}}

\newcommand{\bool}{\mathit{bool}}
\newcommand{\combine}[2]{[{#1}\,{\circledast}\,{#2}]} % find nicer notation?
\newcommand{\dom}{\mathit{dom}}
\newcommand{\U}{\mathrel{\mathsf{U}}} % LTL operator
\newcommand{\R}{\mathrel{\mathsf{R}}} % LTL operator
\newcommand{\G}{\mathsf{G}\xspace}
\newcommand{\F}{\mathsf{F}\xspace}
\newcommand{\X}{\mathsf{X}\xspace}
\newcommand{\Xw}{\mathsf{X}_{\mathsf w}\xspace}
\newcommand{\Xs}{\mathsf{X}\xspace}

\renewcommand{\SS}{\mathcal S}
\newcommand{\PP}{\mathcal P}
\newcommand{\FF}{\mathcal F}
\newcommand{\LL}{\mathcal L}

\newcommand{\TT}{\mathcal T}

\newcommand{\modelsT}{\models}

\newcommand{\ANDOR}[1][\psi]{\mathcal A_{#1}}

\newcommand{\foa}{\mathsf{atoms}} % constraints in symbol
\newcommand{\goto}[1]{\mathrel{\raisebox{-2pt}{$\xrightarrow{#1}$}}}
\newcommand{\gotos}[1]{\mathrel{\raisebox{-2pt}{$\smash{\xrightarrow{#1}}^*$}}}
\newcommand{\assign}[1]{{\mathit{Assign}_{#1}}}
\newcommand{\lab}{\mathit{lab}}
\newcommand{\Win}{\mathit{Win}}
\newcommand{\PreC}{\mathit{Pre}}
\newcommand{\Conf}{\mathcal C}

\newcommand{\remX}{\mathsf{rmX}}
\newcommand{\alphas}{\boldsymbol{\alpha}}

\newcommand{\last}{\mathit{last}}

\newcommand{\pre}[1]{\overline{#1}}

\newcommand{\env}{\mathit{env}}
\newcommand{\seq}{\mathsf{seq}}

\newcommand{\tup}[1]{\langle #1\rangle}

\renewcommand{\phi}{\varphi}

\newcommand{\ag}{{\mathit{ag}}}

\newcommand{\prim}{\mathsf{prm}}
\newcommand{\sub}{\mathsf{sub}}
\newcommand{\decomp}{\mathsf{decomp}}
\newcommand{\xnf}{\mathsf{xnf}}
\newcommand{\ta}{\mathsf{tnps}}

\newcommand{\eval}[1]{\llbracket#1\rrbracket}
\newcommand{\mystrat}{g}
\newcommand{\len}{\mathit{len}}

\newcommand{\LTLfMT}{\ensuremath{\mathsf{LTL_f^{MT}}}\xspace}
\newcommand{\LTLf}{\ensuremath{\mathsf{LTL_f}}\xspace}
\newcommand{\LTL}{\ensuremath{\mathsf{LTL}}\xspace}
\newcommand{\LIA}{\textsf{LIA}\xspace}
\newcommand{\LRA}{\textsf{LRA}\xspace}

\newtheorem{thm}{Theorem}
\newtheorem{example}{Example}
\newtheorem{definition}{Definition}
\newtheorem{lemma}{Lemma}

\newcommand{\lemref}[1]{Lem.~\ref{lem:#1}}

\newcommand{\defref}[1]{Def.~\ref{def:#1}}

\newcommand{\secref}[1]{Sec.~\ref{sec:#1}}

\newcommand{\thmref}[1]{Thm.~\ref{thm:#1}}

\newcommand{\exaref}[1]{Ex.~\ref{exa:#1}}
\newcommand{\figref}[1]{Fig.~\ref{fig:#1}}

\tikzstyle{state}=[draw, circle, inner sep=1.5pt, line width=.7pt, scale=.6]
\tikzstyle{edge}=[draw, ->, line width=.5pt]
\tikzstyle{action}=[scale=.55]
\tikzstyle{caption}=[scale=.9]
\tikzstyle{node} = [draw,rectangle split, rectangle split parts=2,rectangle split horizontal, rectangle split draw splits=true, inner sep=3pt, scale=.65, rounded corners=2pt]
\tikzstyle{goto} = [->]
\tikzstyle{action}=[scale=.6, black]
\tikzstyle{final}=[double]
\tikzstyle{binode} = [draw,rectangle split, rectangle split parts=2,rectangle split horizontal, rectangle split draw splits=true, inner sep=3pt, scale=.65, rounded corners=2pt]

\renewcommand{\todo}[1]{}

\newboolean{isExtended}
\setboolean{isExtended}{true}
\newcommand{\ifextended}[2]{\ifthenelse{\boolean{isExtended}}{#1}{#2}}

\begin{document}
\begin{frontmatter}

\title{First-Order LTLf Synthesis with Lookback \ifextended{\\(Extended Version)}{}}
\author{Sarah Winkler}
\address{Free University of Bozen-Bolzano, Italy}

\begin{abstract}
Reactive synthesis addresses the problem of generating a controller for a temporal specification in an adversarial environment; it was typically studied for \LTL. 
% Traditionally, reactive synthesis has been studied for standard \LTL, where it can be compiled to a two-player game on a DFA encoding the target property. 
Driven by applications ranging from AI to business process management, \LTL modulo first order-theories over finite traces (\LTLfMT) has recently gained traction, where propositional variables in properties are replaced by first-order constraints.
Though reactive synthesis for \LTLf with some first-order features has been addressed, existing work in this direction strongly restricts or excludes the possibility to compare variables across instants, a limitation that severely restricts expressiveness and applicability.

In this work we present a reactive synthesis procedure for \LTLfMT, where properties support \emph{lookback} to model cross-instant comparison of variables.
Our procedure works for full \LTLfMT with lookback, subsuming the fragments of \LTLfMT for which realizability was studied earlier.
However, the setting with cross-instant comparison is inherently highly complex, as realizability is undecidable even over decidable background theories. Hence termination of our approach is in general not guaranteed. Nevertheless, we prove its soundness, and show that it is complete if a bound on the strategy length exists.
Finally, we show that our approach constitutes a decision procedure for several relevant fragments of \LTLfMT, at once re-proving known decidability results and identifying new decidable classes.
\end{abstract}

\end{frontmatter}

\section{Introduction}

Reactive synthesis~\cite{PnueliR89,PnueliR89b} addresses the problem of automatically constructing a system that satisfies a given temporal specification in an adversarial environment. Along with the respective decision problem of realizability, it is most typically considered for \LTL, where the input is an \LTL property and a splitting of the set of propositional variables into two disjoint sets that are respectively controlled by the environment and the agent. The aim is to find a strategy for the agent such that the given property is satisfied, no matter which values are assigned to the environment variables. 
While reactive synthesis was originally considered for \LTL over infinite traces, in a recent line of work, efficient algorithms and tools have been developed that target the synthesis problem for \LTLf, i.e., \LTL over finite traces~\cite{deGiacomoV15,GiacomoFLVX022,XiaoL0SPV21}.

However, propositional \LTL and \LTLf have limited expressivity that inhibits to realistically model many practical problems. This observation led to the development of synthesis techniques for the logic $\LTL_\TT$, where propositional atoms are replaced by atoms from a  first-order theory~\cite{RodriguezS23,RodriguezS24jlamp,RodriguezS24}. These techniques are based on the compilation of a Boolean abstraction from the first-order problem that can be given to a standard realizability tool after solving some first-order satisfiability problems to ensure that the Boolean abstraction is faithful.
Independently, the logic \LTLfMT (\LTLf modulo theories) has gained traction in the last years~\cite{GGG22,GMW24}, as its versatility admits applications in a variety of areas. This includes verification of infinite-state systems using the  MoXI intermediate language~\cite{moxi} whose semantics can be expressed in \LTLfMT, and analysis of business processes~\cite{GMW24}; but \LTLfMT also covers languages used for runtime verification such as LoLA~\cite{DAngeloSSRFSMM05}, and \LTLf with arithmetic constraints considered for monitoring~\cite{FMPW23}.
\LTLfMT also subsumes the logic $\LTL_\TT$, but differs from the latter in that variables can be compared across different time instants.
This feature vastly increases expressivity, is crucial to describe the advancement of a dynamic system, and is present in most languages for reasoning over temporal properties that have some first-order flavour, including MoXI, LoLA, the runtime verification language TeSSLa~\cite{tessla}, and Data Petri nets, a variant of Petri nets popular in the analysis of business processes~\cite{MannhardtLRA16}.  
In this paper we will use the term \emph{lookback} for the feature to compare variables describing the current instant with variables for the previous instant; this is  equivalent in expressivity to comparison with variables that refer to (bounded) future instants~\cite{FMPW23}.
However, cross-state comparisons come at the cost that the satisfiability and realizability problems are undecidable even over a benign theory such as linear arithmetic~\cite{GGG22,BhaskarP24}.
The following simple example illustrates the realizability problem in such a setting:

\begin{example}
\label{exa:intro}
Alice runs a resource-consuming computing task that is supposed to take multiple days in a cloud-based elastic system.
Suppose that at every instant, a rational number $x$ describes the resources allocated by the cloud system (the environment), and a rational number $y$ holds the resource bound requested by Alice (the agent). 
Alice wants to take one day off, so she wants to set a resource bound today that is guaranteed to be larger than the resources required the following day. She estimates that the resource consumption increases from one day to the next by at most 2 units. The aim is to find a strategy for Alice that ensures this property. In \LTLfMT, this can be described as $\G(x \geq 0 \wedge x {- }\pre{x} \leq 2) \to \X (\pre y > x)$, where $\pre x$ and $\pre y$ denote the lookback variables that hold the previous values of $x$ and $y$, respectively.
\end{example}

In this paper we address, to the best of our knowledge for the first time, the reactive synthesis problem for full \LTLfMT with lookback.
Our approach to this end is conceptually simple: we first represent a DFA for the given \LTLfMT property $\psi$ as an AND-OR-graph $\ANDOR$ where edge labels are suitably split into constraints controlled by the environment and constraints controlled by the agent, following~\cite{GiacomoFLVX022}. In a second step, we define for every state in $\ANDOR$ a first-order formula that constitutes a winning condition for the agent. The overall winning condition is defined as a fixpoint which need not exist in general; but if it does, it gives rise to a strategy to synthesise $\psi$. 

In summary, we provide a three-fold contribution:
\begin{inparaenum}[(1)]
\item We present a general reactive synthesis procedure for \LTLfMT properties with lookback.
\item We prove its soundness, and show that it is also complete for \emph{bounded} realizability, i.e., in the case where there is a bound on the length of the strategy.
\item
Our procedure is shown to be a decision procedure for several fragments of \LTLfMT, including the case without lookback; properties over the theory of linear arithmetic where atoms are variable-to-variable/constant comparisons; and properties where variable interaction via lookback is bounded.
\end{inparaenum}

% Notably, our completeness result covers these decidable classes. 
%
The paper is structured as follows: In the remainder of this section we further elaborate on related work. In \secref{background} we recap \LTLfMT and reactive synthesis.
Our synthesis approach is described in \secref{approach}. We outline our decidability results in \secref{decidability}, and conclude in \secref{conclusion}.
Some proofs were moved to an appendix for reasons of space.

\paragraph{Related work.}
Reactive synthesis for \LTL was pioneered in~\cite{PnueliR89,PnueliR89b}, triggering a rich body of literature. Below we focus on LTL over finite traces (\LTLf)~\cite{DeGV13}, and \LTL/\LTLf over first-order theories. Reactive synthesis approach for \LTLf was first studied in~\cite{deGiacomoV15}, constructing a doubly-exponential DFA as a game arena upfront. This led to more efficient techniques that construct the DFA on the fly~\cite{XiaoL0SPV21}, also using knowledge compilation via SDDs for compactness~\cite{GiacomoFLVX022}.
Our approach builds up on their idea to represent a DFA as an AND-OR-graph using sentential decision diagrams (SDDs).
Efficient compositional/hybrid techniques to transform \LTLf to DFAs were also discussed in~\cite{GiacomoF21,BansalLTV20}.
% The latter three approaches are implemented in the tools that compete in the Reactive Synthesis Competition (\url{https://www.syntcomp.org/}).
Reactive synthesis for LTL$_\TT$, i.e., LTL/\LTLf over a first-order theory $\TT$ but without cross-state comparison, was studied in~\cite{RodriguezS23,RodriguezS24,RodriguezS24jlamp}, and realizability proven decidable for theories with a decidable $\exists^*\forall^*$ fragment.
For \LTLf, our \thmref{no:lookback} covers their decidability result.
A few approaches also investigated reactive synthesis over restricted \LTLf modulo theories:
Strategy synthesis has been addressed in~\cite{LFM20} for data-aware dynamic systems where guards are variable-to-variable/constant comparisons over the rationals. Our \thmref{mc} essentially reproves their results.
The synthesis problem has also been considered for constraint \LTL~\cite{BhaskarP24}, which is \LTL over infinite traces where atoms are variable-to-variable/constant comparisons in the theory of linear arithmetic. The authors prove undecidability of the general case, but decidability in 2EXPTIME for theories that enjoy a so-called completion property, and for formulas where cross-state comparisons are restricted to the agent (single-sided games). However, they do not provide a procedure for the general case.

\section{Background}
\label{sec:background}

In this section we recall LTL modulo first-order theories over finite traces (\LTLfMT), reactive synthesis, and progression.

\paragraph{\LTLf modulo theories}
The following definitions are based on~\cite{GGG22}.
We consider a first-order multi-sorted \emph{signature} $\Sigma=\langle\SS,
\PP, \FF, V\rangle$, where
$\SS$ is a set of sorts including $\mathit{bool}$; 
$\PP$ and $\FF$ are sets of predicate and function symbols; and
$V$ is a finite, non-empty set of \emph{data variables}.
% and $W$ is a set of variables disjoint from $V$, which will be used for quantification.
All variables are supposed to have a sort in $\SS$, and each predicate and function
symbol has a type with input sorts from $\SS$ and an output sort in $\SS$.
% We assume that $\PP$ contains equality predicates for all sorts.
%
$\Sigma$-terms $t$ are built as according to the grammar:
\[
t := v \mid \pre v \mid f(t_1, \dots, t_k)
\]
where $v \in V$; $f\in \FF$ has arity $k$ and each $t_i$ is a term of appropriate sort; and $\pre v$ denotes the \emph{lookback} operator applied to $v$. We will use this notation to represent the value of variable $v$ in the previous state (see \defref{semantics} below). The set $\pre V = \{ \pre v \mid v\in V\}$ is called \emph{lookback variables}.
An \emph{atom} has the form $p(t_1,\dots, t_k)$, where $p \in \PP$ is a predicate symbol of arity $k$,
and $t_i$ are terms of suitable sort. 
We consider \LTLfMT properties in negation normal form (NNF):

\begin{definition}
\label{def:syntax}
First-order $\Sigma$-formulas $\phi$ and \LTLfMT properties $\psi$ are defined as follows, where $a$ is an atom:\\[1ex]
$
\begin{array}{@{}r@{\:}l@{}}
\phi &:= \top \mid \bot \mid a \mid \neg a \mid \phi_1 \land \phi_2 \mid \phi_1 \lor \phi_2 \\
% \mid \exists w.\,\phi \mid \forall w.\,\phi \\
\psi &:= \phi \mid \psi_1 \land \psi_2 \mid \psi_1 \lor \psi_2 \mid \X \psi \mid \Xw \psi \mid \psi_1 \U \psi_2 \mid \psi_1 \R \psi_2
\end{array}
$\\[1ex]
The set of all properties $\psi$ is denoted $\LL$.%
\footnote{
Throughout the paper, we will reserve the term \emph{formula} for first-order formulas, and call elements in $\LL$ \emph{properties}.}
\end{definition}

Here $\X$ is the strict next and $\Xw$ the weak next operator.
A first-order formula $\phi$ is 
a \emph{state formula} if all its free variables are in $V$, i.e.,
% all variables in $W$ are quantified and
variables in $\pre V$ do not occur.

A set of $\Sigma$-formulas without free variables is a \emph{$\Sigma$-theory} $\TT$.
In contrast to \cite{GGG22}, we omit quantifiers in $\LL$; % to simplify the presentation;
% NOTE also because splitting of quantified formulas would be difficult
however, theory axioms may have quantifiers, as common in first-order logic.
To model quantified first-order subformulas in $\LL$, one can add extra predicates that are defined by suitable additional axioms (possibly with quantifiers) in the theory.
In the paper we will sometimes refer to common SMT theories~\cite{BarrettSST21}, namely
% : the theory of equality and uninterpreted functions for a given $\Sigma$ (\EUF),
the theories of linear arithmetic over rationals (\LRA) and integers (\LIA).
\smallskip

To define the semantics of first-order formulas, we use the standard notion of a
\emph{$\Sigma$-structure} $M$, which associates each sort $s\in \SS$ with a
domain $s^M$, and each predicate $p\in \PP$ and function symbol $f\in \FF$ with suitable interpretations $p^M$ and $f^M$. The equality predicates have the
usual interpretation as the identity relation. The carrier of $M$,
i.e., the union of all domains of sorts in $\SS$, is denoted by $|M|$.
\todo{why is equality needed}
We call a function $\alpha\colon V \to |M|$ an \emph{assignment} wrt. $M$, where
 all variables are supposed to be mapped to elements of their domain.
% The environment $\eta$ extended with a binding from $u$ to $e$ is denoted $\eta[u \mapsto e]$.
%
A \emph{trace} is a pair $\tau = (M, \langle\alpha_0, \dots,
\alpha_{n-1}\rangle)$ of a $\Sigma$-structure $M$ and a sequence of assignments wrt. $M$.
The length of a trace $\tau$ as above is $|\tau|=n$.
The concatenation of traces $\tau$ and $\tau'$ is denoted $\tau\tau'$.

\begin{example}
\label{exa:trace}
% $\psi = (\G(x \geq 0) \wedge \X\G(x {- }\pre{x} \leq 2)) \to \X (\pre y > x)$
Let $V$ consist of variables $x$ and $y$ of sort \emph{int},
and $M$ be the usual model of the theory of linear arithmetic over the integers (\LIA).
Then, e.g. $(M,\alphas)$ is a trace of length 3 if
$\alphas = \langle \{x \mapsto -1, y \mapsto 0\}, \{x \mapsto 0, y \mapsto 1\},\{x \mapsto 2, y \mapsto 2\}\rangle$.
\end{example}

A term $t$ is \emph{well-defined} at instant $i$ of a trace $\tau$ if $0<i<|\tau|$, or $i=0$ and $t$ does not contain lookback variables. In this case, the evaluation
of $t$ at $i$
% with respect to an environment $\eta$
is denoted $\eval{t}_{\tau}^i$ and defined as
$\eval{v}_{\tau}^i{=}\alpha_i(v)$,
$\eval{\pre v}_{\tau}^i{=}\alpha_{i-1}(v)$, and
$\eval{f(t_1,\dots,t_k)}_{\tau}^i = f^M(\eval{t_1}_{\tau}^i, \dots ,\eval{t_k}_{\tau}^i)$ where $v\,{\in}\,V$.

\begin{definition}[\LTLfMT semantics]
\label{def:semantics}
Satisfaction of a first-order formula $\phi$ in trace $\tau$ at instant $i$, $0{\leq}i\,{<}\,|\tau|$, is denoted
$\tau,i \models \phi$ where:
\[
\begin{array}{@{}r@{\,}ll@{}}
\tau,i \models\: & \top \text{ but }\tau,i \not \models\: \bot \\
\tau,i \models\: & p(t_1, \dots, t_k) & \text{if $t_1, \dots, t_k$ are well-defined and} \\
&&(\eval{t_1}_{\tau}^i, \dots ,\eval{t_k}_{\tau}^i) \in p^M,
\text{ or}\\
&&\text{if some $t_1, \dots, t_k$ is not well-defined}\\
\tau,i \models\: & \neg p(t_1, \dots, t_k) \quad &
\text{if }\tau,i \not\models p(t_1, \dots, t_k)\\
\tau,i \models\: & \phi_1 \land \phi_2 &
\text{if }\tau,i \models \phi_1 \text{ and }\tau,i \models \phi_2\\
\tau,i \models\: & \phi_1 \lor \phi_2 &
\text{if }\tau,i \models\: \phi_1 \text{ or }\tau,i \models \phi_2.
% \\
% \tau,i \models & \exists w.\,\phi &&
% \text{if }\tau,i[w\mapsto e] \models \phi
% \text{ for some } e\in s^M\\
% \tau,i \models & \forall w.\,\phi &&
% \text{if }\tau,i[w\mapsto e] \models \phi
% \text{ for all } e\in s^M.
\end{array}
\]
The satisfaction relation is extended to properties in $\LL$ as follows:

\[
\begin{array}{@{}r@{\,}ll@{}}
% \tau,i \models  & \phi &
% \text{if } \tau,i \models \phi \text{ for some }\eta \\
\tau,i \models  & \psi_1 \land \psi_2 &
\text{if }\tau,i \models  \psi_1 \text{ and }\tau,i \models  \psi_2\\
\tau,i \models  & \psi_1 \lor \psi_2 &
\text{if }\tau,i \models  \psi_1 \text{ or }\tau,i \models  \psi_2\\
\tau,i \models  & \X \psi &
\text{if $i<|\tau|{-}1$ and }\tau,i+1 \models \psi\\
\tau,i \models  & \Xw \psi &
\text{if $i=|\tau|{-}1$ or }\tau,i+1 \models \psi
\end{array}
\]

\[
\begin{array}{@{}r@{\,}ll@{}}
\tau,i \models  & \psi_1 \U\psi_2 &
\text{if there is some $j$, $i \leq j<|\tau|$ such that }\tau,j \models \psi_2 \\
&& \text{and }\tau,k \models \psi_1 \text{ for all }i \leq k <j \\
\tau,i \models  & \psi_1 \R \psi_2 &
\text{if either }\tau,j \models \psi_2 \text{ for all }i \leq j<|\tau|\text{, or there is}\\
&&\text{some $j$, $i \leq j<|\tau|$ such that }\tau,j \models \psi_1 \\
&& \text{and }\tau,k \models \psi_2 \text{ for all }i \leq k \leq j
\end{array}
\]
Finally, $\tau$ \emph{satisfies} $\psi$, denoted by $\tau \models \psi$,  if
$\tau,0 \models \psi$ holds.
\end{definition}

We use the usual shorthands $\F \psi \equiv (\top \U \psi)$ and $\G \psi \equiv (\bot \R \psi)$.
E.g., the trace in \exaref{trace} satisfies
$(y\,{\geq}\,x) \U (x\,{=}\,y)$ and
$\G (x\,{>}\,\pre x)$; as well as $\pre x < x$ and hence $\F(\pre x  < x)$
due to the weak semantics that we assume for variables with lookback, but not $\X \F(\pre x  < x)$.

A first-order formula $\phi$ without $\pre V$ is satisfied by a
$\Sigma$-structure $M$ and assignment $\alpha\colon V \to |M|$,
denoted $M,\alpha \models \phi$, if $(M,\langle \alpha\rangle) \models \phi$,
which corresponds to the usual notion of first-order satisfaction.
If the model is clear from the context, e.g. for LIA or LRA where the model is unique up to isomorphism,
we also simply write $\alpha \models \phi$.
%; if $\phi$ is a sentence, we simply write $M \models \phi$. 
% For a $\Sigma$-structure $M$, we will write $M\in \TT$ to express that $M$ is a model of $\TT$. 
A property $\psi \in \LL$ is \emph{$\TT$-satisfiable} if it is satisfied by some $\tau=(M, \alphas)$ with $M\in\TT$. 
Moreover, two first-order formulas $\phi_1$ and $\phi_2$ are \emph{$\TT$-equivalent}, denoted
$\phi_1 \equiv_\TT \phi_2$, if $\neg(\phi_1 \leftrightarrow \phi_2)$ is not $\TT$-satisfiable.

In this work we frequently need quantification.
A $\Sigma$-theory $\TT$ has \emph{quantifier elimination} (QE) if for any
$\Sigma$-formula $\phi$ there is a quantifier-free formula $\phi'$ that is
$\TT$-equivalent to $\phi$. This holds e.g. for the theories \LIA and \LRA, for which
quantifier elimination procedures
%as well as procedures to check satisfiability of quantified formulas
are implemented in state-of-the-art SMT-solvers~\cite{Z3,cvc5}.
However, SMT-solvers also support general first-order formulas to some extent~\cite{cvc5}, and so do powerful first-order theorem provers~\cite{KovacsV13}, even though satisfiability is in general undecidable.

For a property $\psi \in \LL$, 
$\foa(\psi)$ denotes all atoms of $\psi$, i.e., all occurrences of $a$ by to the grammar of \defref{syntax}.
We also need the set of \emph{top-level non-propositional subproperties} $\ta(\psi)$, given by $\ta(\psi) = \{\psi\}$ if $\psi$ is an atom or rooted by a temporal operator, and $\ta(\psi_1 \vee \psi_2) = \ta(\psi_1 \wedge \psi_2) = \ta(\psi_1) \cup \ta(\psi_2)$ otherwise.

We use a dedicated proposition $\last$ that is supposed to hold only in the final state of a trace.
A property $\psi \in \LL$ is in \emph{next normal form} (XNF) if $\ta(\psi)$ consists only of atoms and properties rooted by $\Xw$ and $\Xs$.
Every $\psi \in \LL$ can be transformed into an equivalent property in XNF obtained by the function $\xnf$, as follows~\cite{GiacomoFLVX022}:
\begin{compactitem}
\item $\xnf(\psi) = \psi$ if $\psi$ is a literal, or rooted by $\Xw$ or $\Xs$,
\item $\xnf(\psi_1 \wedge \psi_2) = \xnf(\psi_1) \wedge \xnf(\psi_2)$,
\item $\xnf(\psi_1 \vee \psi_2) = \xnf(\psi_1) \vee \xnf(\psi_2)$,
\item $\xnf(\psi_1\U\psi_2) = \xnf(\psi_2)\vee(\xnf(\psi_1)\wedge\Xs(\psi_1 \U \psi_2))$, and
\item $\xnf(\psi_1\R\psi_2) = (\xnf(\psi_2)\vee \last)\wedge(\xnf(\psi_1)\vee\Xw(\psi_1 \R\psi_2))$.
\end{compactitem}
For instance, $\G(x\,{\geq}\, 0)$ is equivalent to the property $(x\,{\geq}\, 0) \wedge \Xw(\G(x\,{\geq}\, 0))$ in XNF.

We call a property $\psi \in \LL$ \emph{well-formed} if all atoms in  $\ta(\xnf(\psi))$  are either rooted by $\X$ or $\Xw$, or do not contain lookback. Intuitively, this means that there are no atoms evaluated at the first instant of a trace that refer to variables beyond the start of the trace. For instance,
$\G (x \geq 0 \wedge x - \pre x \leq 2)$ and $(\pre y > x)$ are not well-formed,
but 
$\G(x \geq 0) \wedge \X\G(x {- }\pre{x} \leq 2)$ and $\X (\pre y > x)$ are.

\paragraph{Reactive synthesis.}
In the sequel a signature $\Sigma=\langle\SS,\PP, \FF, V\rangle$ is assumed where
$V$ is decomposed into two fixed disjoint subsets $X$ and $Y$, i.e., $V = X \uplus Y$. Intuitively, $X$ are variables controlled by the environment, and $Y$ variables controlled by the agent.
We consider \emph{environment-first} synthesis~\cite{RodriguezS24}, where
in every instant $i$, first the
environment emits values $\beta_i(x)$ for all $x\in X$, and afterwards the
agent can choose values $\gamma_i(y)$ for all $y\in Y$.
% \footnote{We consider \emph{environment-first} synthesis~\cite{RodriguezS24}; \emph{system-first}~\cite{GiacomoFLVX022,XiaoL0SPV21} is dual.}
Here $\beta_i$ is an assignment with domain $X$, denoted by $\beta_i \in \assign{X}$, and $\gamma_i$ is an assignment with domain $Y$, written as $\gamma_i \in \assign{Y}$.
We write $\pre X= \{ \pre x \mid x\in X\}$ and $\pre Y= \{ \pre y \mid y\in Y\}$ for the respective lookback variables.

For an infinite sequence $\pi$, let $\pi|_k$ denote its prefix of length $k$.
In this paper we consider the reactive synthesis problem for \LTLfMT:

\begin{definition}[Synthesis problem]
\label{def:synthesis}
The \emph{synthesis problem} for $\psi\in \LL$ with respect to $(X,Y)$ asks to find a strategy $g\colon \assign{X}^* \to \assign{Y}$ such that for every infinite sequence of assignments $\pi =\beta_0, \beta_1, \beta_2, \dots \in \assign{X}^*$, there is some $k \geq 0$ such that the trace $\tau = \langle\beta_0 \cup g(\pi|_1), \beta_1 \cup g(\pi|_2), \dots, \beta_k \cup g(\pi|_{k+1})\rangle$ satisfies $\tau \models \psi$.

The \emph{realizability problem} is the related decision problem, asking whether $k$ and suitable assignments $g(\pi|_i)$, $0< i \leq k+1$ exist.
\end{definition}

\noindent
We illustrate the synthesis problem for some simple properties.

\begin{example}
 Suppose  $x\in X$ and $y\in Y$.
\begin{compactitem}
 \item All of $x=y$, $\G (x=y)$, and $\G (y=\pre x)$ are realizable over \LIA or \LRA because the agent chooses $y$ after the environment fixes $x$.
 Also $\pre y=x$ is realizable due to the weak semantics of lookback, but not $\X (\pre y=x)$, as the agent cannot know the next value of $x$.
 \item $\F (x=1) \to \X\F (\pre y=x)$ is realizable by always setting $y$ to 1 (note that the agent determines the trace length).
 \item The running example from~\cite{RodriguezS24}, $\G ((x\,{<}\,2 \to \X(y\,{>}\,1))\wedge (x\,{\geq}\,2 \to y\,{<}\,x))$ is not realizable over \LIA because the two conjuncts can impose contradicting requirements e.g. if $\beta_0(x)=1$ and $\beta_1(x)=2$, but it is realizable over \LRA.
\end{compactitem}
\end{example}

The synthesis problem for \LTLfMT is undecidable in general, more precisely, even if $\TT$ is the decidable theory \LIA, and atoms are restricted to variable-to-variable/constant comparisons~\cite[Thm.~1]{BhaskarP24}.

In the context of synthesis, it is often desirable to find a strategy realizing a property $\psi \in \LL$ \emph{under the assumption} that the environment behaves in a way that is expressed as a property $\chi\in \LL$ over variables $X \cup \pre X$. For safety or stability assumptions, synthesis with assumptions can be reduced to standard synthesis (\defref{synthesis}), by searching a strategy for $\chi \to \psi$ ~\cite{GiacomoSVZ20}.
We exploit this in the next example.

\begin{example}
\label{exa:assumption}
%  Let again $x\in X$ and $y\in Y$, and the theory be \LIA.
% \begin{compactitem}
%  \item Suppose the $x$-values chosen by the environment are weakly monotonically increasing, expressed by $A = \G (x \geq \pre x)$, and consider the property $\psi = (\G (y \geq x) \wedge \F (y=2))$. Then $A \to \psi$ is not realizable because the first value of $x$ can be greater than 2.
%  \item
Consider again \exaref{intro}.
As an assumption on the environment, we suppose that the values of $x$ are non-negative and increase by at most 2 in every step, as expressed by $\chi = \G(x \geq 0 \wedge x {- }\pre{x} \leq 2)$.
The desired property is given by $\psi = \X (\pre y > x)$. Then $\chi \to \psi$ is realizable, because using the knowledge about the bounded increase of $x$, $y$ can be chosen greater than $x + 2$ in the first instant.
% \end{compactitem}
\end{example}

\paragraph{\LTLf progression.}
We consider progression over finite traces~\cite{GiacomoFLVX022} to (implicitly) construct a DFA for a property $\psi\in \LL$. For now we ignore the first-order nature of atoms in \LTLfMT properties.

\begin{definition}[\LTLf Progression]
\label{def:progression}
For $\psi \in \LL$ and $A$ a set of atoms,
progression $\psi^+(A)$ of $\psi$ with respect to $A$ is defined as follows:
\begin{compactitem}
% \item 
% $\top^+(A) = \top$ and $\bot^+(A) = \bot$,
\item
if $\psi$ is an atom, $\psi^+(A){=}\top$ if $\psi \in A$ and $\psi^+(A){=}\bot$ otherwise;
\item 
if $\psi = \neg a$ for an atom $a$ then $\psi^+(A) = \bot$ if $a\in A$ and $\psi^+(A) = \top$ otherwise;
\item 
if $\psi = \psi_1 \wedge \psi_2$ then $\psi^+(A) = \psi_1^+(A) \wedge \psi_2^+(A)$;
\item 
if $\psi = \psi_1 \vee \psi_2$ then $\psi^+(A) = \psi_1^+(A) \vee \psi_2^+(A)$;
\item 
if $\psi = \X \psi_1$ then $\psi^+(A) = \psi_1 \wedge \neg\last$;
\item 
if $\psi = \Xw \psi_1$ then $\psi^+(A) = \psi_1 \vee \last$;
\item 
if $\psi = \psi_1 \U \psi_2$ then $\psi^+(A) = \psi_2^+(A) \vee (\psi_1^+(A) \wedge (\X\psi)^+(A))$;
\item 
if $\psi = \psi_1\,{\R}\, \psi_2$ then $\psi^+(A) = \psi_2^+(A) \wedge (\psi_1^+(A) \vee (\Xw\psi)^+(A))$.
\end{compactitem}
\end{definition}

Let a sequence of sets of atoms $\sigma = \tup{A_0, A_1, \dots, A_{n-1}}$  be \emph{over} a set of atoms $C$ if $A_i \subseteq C$ for all $i$, $0\leq i < n$.
Progression is as usual generalized to sequences, by setting $\psi^+(\epsilon) = \psi$ for the empty sequence and $\psi^+(A \rho) = (\psi^+(A))^+(\rho)$ otherwise.
The following are the key properties of progression~\cite[Lem. 1 and 3]{GiacomoFLVX022}.

\begin{lemma}
\label{lem:progression}
Let $\psi \in \LL$ and $\sigma = \tup{A_0,A_1, \dots, A_{n-1}}$ a sequence of sets of atoms over $\foa(\psi)$.
Then 
\begin{compactenum}[(1)]
\item $\sigma,i \models \psi$ iff $\sigma,i+1 \models \psi^+(A_i)$ for all $i$, $0 \leq i < n$, and
\item $\sigma \models \psi$ iff $\psi^+(\sigma) \equiv \top$.
\end{compactenum}
\end{lemma}

Here we use the notion of a sequence of atom sets $\sigma = \tup{A_0, A_1, \dots, A_{n-1}}$ satisfying a propositional/standard \LTLf property. The semantics, however, coincides with that of \defref{semantics} where atoms are considered variables of sort $\bool$, and $\sigma$ is read as a trace where at instant $i$, an atom $a$ is true iff $a\in A_i$.

% Let $\sim$ denote propositional equivalence of formulas in $\LL$, i.e., the equivalence relation obtained by applying boolean equivalences.
% Formula progression of $\sim$-equivalent formulas yields $\sim$-equivalent formulas~\cite[Lem. 2]{GiacomoFLVX022}, below we will use this fact implicitly.

\begin{example}
\label{exa:progression}
Consider $\psi=y\,{\geq}\,x \U x\,{=}\,y$ and the sequence of atom sets $\sigma=\tup{\{y\,{\geq}\,x\}, \{y\,{\geq}\,x\}, \{y\,{\geq}\,x, x\,{=}\,y\}}$.
As $\psi^+(\{y\,{\geq}\,x\})\equiv \psi \wedge \neg \last$ and
$\psi^+(\{y\,{\geq}\,x, x\,{=}\,y\})\equiv \top$, we have $\psi^+(\sigma) = \top$.
\end{example}

Given a property $\psi$, one can construct a DFA for $\psi$ by taking as set of states all properties $\psi'$ such that $\psi'=\psi^+(\sigma)$ for some sequence $\sigma$ over $\foa(\psi)$; and with a transition $\psi_1 \goto{A} \psi_2$ iff $\psi_1^+(A) = \psi_2$~\cite{GiacomoFLVX022}. However, we do not reason over DFAs explicitly here.

For assignments $\alpha$, $\alpha'$ with domain $V$, let the assignment $\combine{\alpha}{\alpha'}$ with domain $\pre V{\cup}V$ be defined as 
$\combine{\alpha}{\alpha'}(\pre{v}) = \alpha(v)$ and 
$\combine{\alpha}{\alpha'}(v)\,{=}\,\alpha'(v)$, for all $v\,{\in}\, V$.
A sequence of formulas $\sigma{=}\tup{w_0, \dots, w_{n-1}}$ is \emph{well-formed} if $w_0$ does not contain $\pre V$.

\begin{definition}
A trace $\tau=(M,\tup{\alpha_0, \alpha_1, \dots, \alpha_{n-1}})$ is \emph{consistent} with
a well-formed sequence of formulas $\rho = \tup{w_0,w_1, \dots, w_{n-1}}$
if $M,\alpha_0 \models  w_0$ and
$M,\combine{\alpha_{i-1}}{\alpha_{i}} \models w_i$
for all $i$, $0 < i < n$.
% and $\tau$ is consistent with a 
% sequence of atom sets $\sigma = \tup{A_0, \dots, A_{n-1}}$ if it is consistent with $\tup{\bigwedge A_0, \dots, \bigwedge A_{n-1}}$.
\end{definition}

For a trace $\tau$ and $\psi\in \LL$, let the sequence of atom sets \emph{corresponding to} $\tau$, denoted $\seq_\psi(\tau)$, be defined as follows: for $\tau = (M,\tup{\alpha_0, \dots, \alpha_{n-1}})$ let $\seq_\psi(\tau) = \tup{A_0, \dots, A_{n-1}}$ such that $A_0$ is the set of $a\in \foa(\psi)$ such that $a$ does not contain $\pre V$ and $M,\alpha_0 \models a$; and
$A_i$ the set of $a\in \foa(\psi)$ such that $M,\combine{\alpha_{i-1}}{\alpha_i} \models a$, for $0 < i < n$.
The following result follows from \lemref{progression} and the definition of $\seq_\psi(\tau)$.
% for the proof see \ifextended{the appendix}{\cite{extended}}:
\todo{add proof?}

\begin{restatable}{thm}{thmbasic}
\label{thm:progression:satisfaction}
Let $\psi \in \LL$ be well-formed and $\sigma = \seq_\psi(\tau)$ be the sequence of atom sets for $\tau$.
Then $\psi^+(\sigma) \equiv \top$ iff $\tau \models \psi$.
\end{restatable}

E.g. for the trace from \exaref{trace} and $\psi$ as in Ex.\:\ref{exa:progression} the corresponding sequence of atom sets is $\sigma$ from Ex.\:\ref{exa:progression}, and we have indeed $\psi^+(\sigma) \,{\equiv}\,\top$ and $\tau\,{\models}\,\psi$.

% \todo{are DFAs needed?}
% In the sequel, for $\psi \in \LL$, let $\Sigma_\psi = 2^{\foa(\psi)^{\pm}}$.
% Let $\Reach(\psi)$ be the set of all $\psi^+(\rho)$, for $\rho$ a sequence of constraint sets over $\Sigma_\psi$.
% A DFA $\DFA$ for $\psi$ can now be defined as follows:
% $\DFA = (Q,\Sigma_\psi,\delta,q_0,Q_F)$ where $Q$ is the set of $\sim$-equivalence classes of $\Reach(\psi)$, the initial state is $[\psi]_\sim$, $Q_F= [\top]_\sim$, and $\delta$ is given by $\delta([\chi]_\sim, w) = [\chi^+(w)]_\sim$ for all $\chi \in \Reach(\psi)$.
% Then we have~\cite[Thm. 2]{GiacomoFLVX022}:
% 
% \begin{thm}
% \label{thm:DFA}
% Let $\psi \in \LL$ and $\rho$ a sequence of constraint sets over $\foa(\psi)^\pm$. Then $\rho$ is accepted by $\DFA$ iff $\hat\delta(q_0,\rho) \in Q_F$.
% \end{thm}
% 
% % A word $w = \varsigma_0, \varsigma_1, \dots, \varsigma_{n-1}\in \Sigma_\psi^*$ is \emph{well-formed} if there are no (negated) constraints in $\varsigma_0$ that mention $\pre V$.
% 
% We get the following as a corollary of Thm.~\ref{thm:DFA} similar to~\cite{FMPW23}:
% 
% \begin{corollary}
% Given $\psi\inn\LL$, a trace $\tau$ and a well-formed word $w\inn \Sigma_\psi^+$ consistent with $\tau$, $\DFA$ accepts $w$ iff $\tau \models \psi$. 
% \end{corollary}

\paragraph{Boolean function decomposition.}
Following~\cite{GiacomoFLVX022}, we will sometimes consider a property $\psi \in \LL$ in XNF as a boolean formula where all subformulas rooted by a temporal operator are considered as propositions, and apply boolean function decomposition.
For a property $\psi \in \LL$ and $C \subseteq \foa(\psi)$, we suppose that $\decomp$ is a procedure such that $\decomp(\psi, C)=\{(\prim_1,\sub_1),$ $\dots, (\prim_k,\sub_k)\}$ for some $k\geq 1$ such that $\prim_i$ and $\sub_i$ are properties in $\LL$ (or equivalently, boolean formulas in the above sense) satisfying the following conditions:
\begin{compactenum}[(1)]
\item
$\psi \equiv \bigvee_{i=1}^k (\prim_i \wedge \sub_i)$, i.e., $\decomp(\psi, C)$ represents $\psi$,
\item
$\prim_i$ has no temporal operators,
all atoms of $\prim_i$ are in $C$ but no atom of $\sub_i$ is in $C$ (i.e., $\ta(\sub_i) \cap C = \emptyset$), for all $i$,
\item 
$\prim_1, \dots, \prim_k$ are disjoint, i.e. $\prim_i \wedge \prim_j \equiv \bot$ for all $i \neq j$,
\item 
$\prim_1, \dots, \prim_k$ are covering, i.e. $\bigvee_{i=1}^k \prim_i \equiv \top$, and
% \item 
% $\prim_i \wedge \sub_i$ is satisfiable for all $i$,
% and optionally 
\item the representation is compressed, i.e., $\sub_i \neq \sub_j$ for $i \neq j$.
\end{compactenum}
% Compression as in Item (5) is in fact optional, and only desirable for efficiency.
The decomposition $\decomp(\psi, C)$ can be found using SDDs~\cite{Darwiche11}.
Note that at this point we ignore the first-order nature of atoms, so $\equiv$ is to be read as equivalence in propositional logic. Thus, some $\prim_i\wedge \sub_i$ may be satisfiable when considering $\foa(\psi)$ as opaque propositional atoms, but unsatisfiable in $\TT$.

E.g., let $\psi = (x < 0) \vee \X \F(x\,{<}\,0 \vee x {- }\pre{x}\,{>}\, 2) \vee \X (\pre y > x)$, which is in XNF. We abstract $\psi$ to the propositional formula $a \vee b \vee c$ for $a = x\,{<}\,0$, $b = \X \F(x\,{<}\,0 \vee x {- }\pre{x}\,{>}\, 2)$, and $c = \X (\pre y > x)$.
Decomposition for $C = \{a\}$ yields
$\decomp(\psi, C) = \{(a, \top), (\neg a, b \vee c)\}$.

\section{Reactive synthesis approach}
\label{sec:approach}

Throughout this section, we assume a property $\psi\in \LL$ that is well-formed.
We use the idea from~\cite{GiacomoFLVX022} to split atom sets into those atoms 
that can be controlled by the agent, and those controlled by the environment.
However, due to the presence of lookback, the separation is a bit more subtle.
By \defref{syntax}, the atoms of $\psi$ have free variables $V\cup \pre V$, i.e., $X\cup Y \cup \pre X \cup \pre Y$.
Let
$C_{\env}\subseteq \foa(\psi)$ be the set of all atoms in $\foa(\psi)$ that do not mention a variable in $Y$, but only variables in $X \cup \pre X \cup \pre Y$.
The values of these variables, and the atoms that depend on them, 
cannot be influenced by the agent at the current instant, as they were either already determined earlier, or are fixed by the environment. We set $C_\ag = \foa(\psi) \setminus C_\env$, the set of atoms that can be influenced by the agent.

Next, we define an AND-OR-graph $\ANDOR$ for $\psi$, which can be seen as a DFA for $\psi$ where edge labels are split according to atoms determined by the environment (and earlier valuations), and atoms controlled by the agent. At choice points for the former we put an AND-node, while at choice points for the agent an OR-node occurs. Each AND-node is moreover labelled by an \LTLfMT property, which intuitively expresses what remains to be satisfied to guarantee $\psi$.

\begin{definition}[AND-OR-graph]
\label{def:andor}
For $\psi \in \LL$ in XNF,
the AND-OR-graph $\ANDOR[\psi] = (S_\land, S_\lor, E, s_0,\lab)$ is given by a set of 
AND-nodes $S_\land$, a set of OR-nodes $S_\lor$, edges $E \subseteq (S_\land \times S_\lor) \cup (S_\lor \times S_\land)$, an initial node $s_0\in S_\land$, and a labelling $\lab\colon S_\land \to \LL$ mapping AND-nodes to \LTLfMT properties.
% Each AND-node $(\psi',\phi) \in S_\land$ consists of some $\psi' \in \LL$ and a state formula $\phi$. 
$\ANDOR[\psi]$ is inductively defined as follows:
\begin{compactenum}
\item[(1)]
The initial node $s_0$ has label $\lab(s_0) = \psi$.
\end{compactenum}
For each AND-node $s$ with $\lab(s)=\psi'$ such that $\psi' \neq \top$, and 
\[\decomp(\xnf(\psi'), C_\env)=\{(\prim_1,\sub_1), \dots,(\prim_k,\sub_k)\}
\]
for every $i$, $1 \leq i \leq k$,  
$\ANDOR[\psi]$ contains:
\begin{compactenum}
\item[(2)]
an OR-node $n_i$ with edge $s \goto{\prim_i} n_i$, and
\item[(3)] 
for the subsequent decomposition
$\decomp(\xnf(\sub_i), C_\ag)=$ $\{(\prim_{i,1},\sub_{i,1}),\dots,(\prim_{i,m},\sub_{i,m})\}$ 
for each $j$, $1\leq j \leq m$ 
% such that $\phi_{i,j}:=\update(\phi, \prim_i \wedge \prim_{i,j})$ 
there are an AND-node $s_{i,j}$ with $\lab(s_{i,j}) = \xnf(\remX(\sub_{i,j}))$
and an edge $\smash{n_i \goto{\sub_{i,j}} s_{i,j}}$.
\end{compactenum}
\end{definition}

Here since all atoms were already split off by decompositions, all top-level atoms in $\sub_{i,j}$ are rooted by $\X$ or $\Xw$, and
we use the function $\remX$ that strips off all $\X$ and $\Xw$ operators, defined as~\cite{GiacomoFLVX022}:\\
% $\remX(\bot)=\remX(\last)=\bot$, $\remX(\top)=\remX(\neg\last)=\top$, $\remX(\psi_1 \vee\psi_2) = \remX(\psi_1) \vee\remX(\psi_2)$, $\remX(\psi_1 \wedge\psi_2) = \remX(\psi_1) \wedge\remX(\psi_2)$,
% $\remX(\X \psi) = \psi \wedge \neg\last$, and $\remX(\Xw \psi) = \psi \vee\last$~\cite{GiacomoFLVX022}.
$\begin{array}{@{}r@{\:}lr@{\:}l@{}}
\remX(\bot)&=\bot &
\remX(\top) &=\top \\
\remX(\last)&=\bot &
\remX(\neg\last)&=\top \\
\remX(\X \psi) &= \psi \wedge \neg\last &
\remX(\psi_1 \vee\psi_2) &= \remX(\psi_1) \vee\remX(\psi_2) \\
\remX(\Xw \psi) &= \psi \vee\last &
\remX(\psi_1 \wedge\psi_2) &= \remX(\psi_1) \wedge\remX(\psi_2)
\end{array}$

We assume that node with the same property label are shared. Then $\ANDOR$ is always finite, because essentially it is a DFA where edge labels are split according to $C_\env$ and $C_\ag$, similar as in~\cite{GiacomoFLVX022}.
An AND-node labeled $\top$ is called \emph{final}, it is unique if it exists.

The following example illustrates the construction of $\ANDOR$.

\newcommand{\psione}{\psi_2}
\newcommand{\psitwo}{\psi_1}
\begin{example}
\label{exa:andor}
We consider
$\psi = (\G(x \geq 0) \wedge \Xw\G(x {- }\pre{x} \leq 2)) \to \X (\pre y > x)$ from \exaref{assumption}, which is equivalent to the property
$\F(x < 0) \vee \X\F(x {- }\pre{x} > 2) \vee \X (\pre y > x)$ in negation normal form.
This corresponds to the AND-OR-graph $\ANDOR$ in \figref{graph}, constructed as follows:
\begin{compactitem}
\item
 $\psi$ is equivalent to
$\psi'=(x < 0) \vee \X \psione \vee \X (\pre y > x)$ in XNF,
for $\psione=\F(x\,{<}\,0 \vee x {- }\pre{x}\,{>}\, 2)$.
The decomposition according to $C_\env$ yields
$\decomp(\psi', C_{\env}) = \{(x \geq 0, \X \psione \vee \X (\pre y > x)), (x\,{<}\,0 , \top)\}$, so we construct two OR-nodes with edges from $\psi$.
The decomposition of $\X \psione \vee \X (\pre y > x)$ according to $C_\ag$ is simply
$\decomp(\X \psione \vee \X (\pre y > x), C_\ag) =\{
( \top,\X \psione \vee \X (\pre y > x))\}$, and $\remX(\X \psione \vee \X (\pre y > x)) = \psione \vee (\pre y > x) =: \psitwo$.
and trivially, we have $\decomp(\top, C_\ag) =\{(\top,\top)\}$, so we obtain new AND-nodes $\top$ and $\psitwo$.
\item 
Next, $\psitwo$ is equivalent to the XNF
$\psitwo'=(x < 0) \vee (x {- }\pre{x} > 2) \vee \X \psione \vee (\pre y > x)$. we get
$\decomp(\psitwo', C_{\env}) = \{(x\,{<}\,0 \vee x{- }\pre{x}\,{>}\, 2, \top), (x\,{\geq}\,0 \wedge x {- }\pre{x}\,{\leq}\,2, \X \psione \vee (\pre y > x))\}$, giving rise to two OR-nodes.
Next, for $\psi_3 := \X \psione \vee (\pre y > x)$ we have
$\decomp(\psi_3, C_\ag ) = \{(\pre y > x, \top), (\pre y\,{\leq}\,x, \X \psione )\}$, and as $\remX(\X\psione) = \psione$, we add an AND-node labelled $\psione$.
\item 
$\psione$ is equivalent to $\psione' = (x < 0) \vee (x {- }\pre{x} > 2) \vee \X \psione $ in XNF, and
$\decomp(\xnf(\psione), C_{\env}) = \{((x < 0) \vee (x {- }\pre{x}\,{>}2), \top), (x\,{\geq}\, 0 \wedge x {- }\pre{x}\,{\leq}\, 2, \X \psione)\}$
and
$\decomp(\xnf(\Xs \psione), C_\ag) = \{(\top, \Xs \psione)\}$.
\end{compactitem}
\begin{figure}
\centering
\begin{tikzpicture}[node distance=25mm]
 \node[state] (A) {$\psi$};
 \node[state, right of =A,yshift=7mm,xshift=-2mm] (A1) {$\vee$};
 \node[state, right of =A,yshift=-7mm,xshift=-2mm] (A2) {$\vee$};
 \node[state, right of=A2,xshift=-8mm] (B) {$\psitwo$};
 \node[state, right of =B,yshift=14mm,xshift=12mm] (B1) {$\vee$};
 \node[state, right of =B,yshift=0mm,xshift=12mm] (B2) {$\vee$};
 \node[state, right of=B1, final] (T) {$\top$};
 \node[state, right of=B2] (C) {$\psione$};
 \node[state, right of=C,yshift=0mm,xshift=5mm] (C1) {$\vee$};
 \node[state, right of=C,yshift=14mm, xshift=5mm] (C2) {$\vee$};
\draw[edge] (A) -- node[action, above,sloped] {$x\,{<}\,0$} (A1);
\draw[edge] (A) -- node[action, below,xshift=3mm,sloped] {$x\,{\geq}\,0$} (A2);
\draw[edge, bend left=10] (A1) to node[action, above] {$\top$} (T);
\draw[edge] (A2) -- node[action, below] {$\top$} (B);
\draw[edge] (B) -- node[action, above,sloped] {$x\,{<}\,0 \vee x{-}\pre x\,{>}\, 2$} (B1);
\draw[edge] (B) -- node[action, below,sloped]{$x\,{\geq}\,0 \wedge x{-}\pre x\,{\leq}\, 2$}  (B2);
\draw[edge] (B1) -- node[action, below, near start] {$\top$} (T);
\draw[edge] (B2) --  node[action, above, sloped] {$ \pre y\,{>}\,x$} (T);
\draw[edge] (B2) --  node[action, below] {$ \pre y\,{\leq}\,x$} (C);
\draw[->] (C1) to[bend right=10]  node[action, above] {$\top$}  (C);
\draw[->] (C) to[bend right=10]  node[action, below, sloped] {$x\,{\geq}\,0 \wedge x{-}\pre x\,{\leq}\, 2$}  (C1);
\draw[->] (C2) to  node[action, above] {$\top$}  (T);
\draw[->] (C) to  node[action, above, sloped, xshift=-3mm] {$x{<}0 \vee x{-}\pre x{>} 2$}  (C2);
\end{tikzpicture}
 \caption{AND-OR graph for $\F(x < 0) \vee \X\F(x {- }\pre{x} > 2) \vee \X (\pre y > x)$}
 \label{fig:graph}
\end{figure}
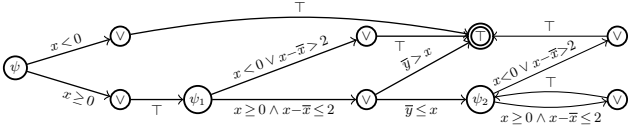
\end{example}
\medskip

Given $\psi \in \LL$ and an AND-OR-graph  $\ANDOR$ for $\psi$ with initial node $s_0$,
if there is a path of the form
\begin{equation}
\label{eq:path}
s_0 \goto{w_1} u_1 \goto{\varsigma_1} s_1 \goto{w_2} u_2 \goto{\varsigma_2}  \dots \goto{w_\ell} u_\ell \goto{\varsigma_\ell} s_\ell
\end{equation}
with AND-nodes $s_i$ and OR-nodes $u_i$,
 we write $s_0 \gotos{\rho} s_\ell$
where $\rho$ is the sequence of formulas over $\foa(\psi)$ given by 
$\rho = \langle w_1 \wedge \varsigma_1, w_2 \wedge \varsigma_2, \dots, w_\ell \wedge \varsigma_\ell\rangle$.
% have $\ta(w_k) \subseteq C_\env^\pm$ and $\ta(\varsigma_k) \subseteq C_\ag^\pm$, 
Moreover, we say that a sequence of atom sets $\sigma = \tup{A_1, \dots, A_\ell}$ such that $A_i\subseteq \foa(\psi)$ \emph{satisfies} $\rho$ if $A_i$ satisfies $w_i \wedge \varsigma_i$ in a propositional sense, for all $1 \leq i \leq \ell$.
The following lemma links the AND-OR-graph to progression:

\begin{restatable}{lemma}{lemmaandor}
\label{lem:andor}
Let $\psi$ be well-formed and $\ANDOR{=}(S_\land, S_\lor, E, s_0,\lab)$.
\begin{compactenum}[(1)]
\item
If %$\rho$ is a sequence of formulas %with atoms in $\foa(\psi)$
%such that
$s_0 \gotos{\rho} s_\ell$ and $s_\ell$ is final
then for every sequence of atom sets $\sigma$ that satisfies $\rho$, it holds that
$\psi^+(\sigma) = \top$.
\item
If $\psi^+(\sigma) = \top$ for some sequence of atom sets 
$\sigma$ over $\foa(\psi)$
then there is %a unique final node $s_\ell$ in $\ANDOR$,
a unique sequence of formulas $\rho$
such that $\sigma$ satisfies $\rho$ and $\ANDOR$ admits a path
$s_0 \gotos{\rho} s_\ell$, and $s_\ell$ is final.
\end{compactenum}
\end{restatable}

It is proven via induction on $\rho$ resp. $\sigma$, using properties of SDD decompositions. The proof can be found in the appendix.

Below, let $\ANDOR[\psi] = (S_\land, S_\lor, E, s_0, \lab)$ be an AND-OR-graph for $\psi$.
We call a \emph{configuration} of $\ANDOR$ a pair $(s, \alpha)$ such that $s \in S_\land$ and $\alpha$ is a state variable assignment.
% The set of configurations is typically infinite.
To represent a (possibly infinite) set of configurations, let a \emph{configuration mapping} be a function $\Conf$ that assigns every node in $S_\land$ a state formula. Intuitively, $\Conf$ represents the set of configurations
$\{(s,\alpha) \mid s \in S_\land \text{ and } \alpha \models \Conf(s)\}$.

Below, the \emph{controllable preimage} of a configuration mapping $\Conf$ is defined as another configuration mapping that represents the set of all configurations in which, for all choices of values of $X$, there is a choice of values of $Y$ such that $\ANDOR$ progresses to a configuration of $\Conf$.
% This is inspired by the approach by de Giacomo and Vardi~\cite{deGiacomoV15}, but in the first-order setting, the notions are more intricate.
To present the formal definition, we need some notation to rename variables in a formula.
For a state formula $\phi$, we write $\phi[V/\pre V]$ for the formula obtained from $\phi$ by replacing $v$ by $\pre v$ for all $v\in V$.
Similarly, for a formula $\phi'$ with free variables $\pre V$, we write $\phi'[\pre V/V]$ for the formula obtained from $\phi'$ by replacing $\pre v$ by $v$ for all $v\in V$.

\begin{definition}[Controllable preimage]
Given a configuration mapping $\Conf$, the \emph{controllable preimage} $\PreC_\Conf$ of $\Conf$ is defined as follows:
for an AND-node $s \in S_\land$ of $\ANDOR$ with children $t_1, \dots, t_m$, and edges $s \goto{w_i} t_i$ for $1\leq i \leq m$, 
\[
\PreC_\Conf(s) = \bigwedge_{i=1}^m (\forall X.\: (w_i \to \PreC_\Conf'(t_i)))[\pre V/V]
\]
where for $t \in S_\vee$ whose edges to children are $t \goto{\varsigma_j} s_j$ for $1\,{\leq}\,j\,{\leq}\,n$, 
\[
\PreC_\Conf'(t) = \bigvee_{j=1}^n \exists Y.\: (\varsigma_j \wedge \Conf(s_j)).
\]
\end{definition}

Based on the controllable preimage, we can define a family of configuration mappings $\Win_i$ that intuitively capture those configurations where the agent has a winning strategy in at most $i$ steps:

\begin{definition}
Given an AND-OR-graph $\ANDOR$ with initial node $s_0$,
let the configuration mappings $\Win_i$ be defined as
$\Win_0(s)=\top$ if $s$ is final and $\bot$ otherwise; and 
$\Win_{i+1}(s) = \Win_{i}(s) \vee \PreC_{\Win_{i}}(s)$.
% \begin{xalignat*}{2}
% \Win_0(s) &= \begin{cases} \top & \text{if $s$ is final}\\ \bot & \text{otherwise}\end{cases} \\
% \Win_{i+1}(s) &= \Win_{i}(s) \vee \PreC_{\Win_{i}}(s)
% \end{xalignat*}
Finally, $\Win(\ANDOR) = \bigvee_{i\geq 0} \Win_{i}(s_0)$.
\end{definition}

The fixpoint $\Win(\ANDOR)$ need not exist in general, but we will show that it does in a variety of relevant cases.

The concepts of controllable preimage and winning strategy are inspired by~\cite{deGiacomoV15}, but there the controllable preimage and $\Win_i$ are sets of DFA states, and since this set is finite the fixpoint always exists. In the first-order setting, the notion of a state gets replaced by that of a configuration, and one needs to reason on configuration mappings as the number of configurations is in general infinite.

\begin{example}
\label{exa:preimage}
Consider the AND-OR-graph in Ex.~\ref{exa:andor}.
Let $s_0$ be the initial node labeled $\psi$, $s_1$ the AND-node labelled $\psitwo$, and $s_2$ the one labelled $\psione$.
We have 
\begin{compactitem}
\item
$\Win_0(s_f)\,{=}\,\top$ for the final node $s_f$, and $\Win_0(s){=}\bot$ otherwise.
\item
Next, $\Win_1(s_0)=\Win_1(s_2) = \bot$ but\\
$
\Win_1(s_1){=}\left (
\begin{array}{@{}l@{}l@{}l@{}}
&\forall x.& 
 (x \geq 0 \wedge x {-} \pre x\,{\leq}\, 2 \to \\
 &&(\exists y.\pre y\,{>}\, x \wedge \top) \vee (\exists y.\pre y\,{\leq}\, x \wedge \bot)) \wedge {} \\
 &\forall x. &(x\,{<}\,0 \vee x {-} \pre x\,{>}\, 2 \to \exists y.\top)
\end{array}\right )\!\![x,y]
$\\
which simplifies to $\Win_1(s_1) = y > x+2 \vee x < -2$.
\item
Now $\Win_2(s_2)$ is still $\bot$, $\Win_2(s_1) =\Win_1(s_1)$, and \\
$
\Win_2(s_0) = \left (
\begin{array}{@{}l@{}l@{}l@{}}
\forall x. (x\,{<}\,0 \to \exists y. \top) \wedge {} \\\forall x. (x\,{\geq}\,0 \to \exists y. (y > x+2 \vee x\,{<}\,{-}2)) 
\end{array} \right)[x,y]
$\\
which simplifies to $\top$.
\end{compactitem}
At this point, $\Win(\ANDOR)=\top$ can be concluded.
\end{example}

Below we use the following notion of a restricted form of realizability where the length of the satisfying trace is bounded:

\begin{definition}
\label{def:bounded:realizable}
A property $\psi \in \LL$ is \emph{boundedly realizable} if
% there is some $m$ such that for all $\pi=\beta_1,\beta_2,\beta_3,\dots$ there are some $i \leq m$ and assignments $\gamma_1, \gamma_2, \dots, \gamma_i$ such that the trace $\tau = \langle \beta_1 \cup \gamma_1, \dots, \beta_i \cup \gamma_i\rangle$ satisfies $\tau \models \psi$.
there are some strategy $h:\assign{X}^* \to \assign{Y}$ and some $m>0$ such that for all $\pi=\beta_0,\beta_1,\beta_2, \dots$ there is some $i\leq m$ such that the trace $\tau = \langle\beta_0 \cup h(\pi|_1), \dots ,\beta_i \cup h(\pi|_{i+1})\rangle$ satisfies $\tau \models \psi$.
\end{definition}

We call $\Win(\ANDOR)$ \emph{satisfiable} we assume implicitly that the fixpoint is well-defined and $\Win(\ANDOR)$ is a satisfiable formula.
Below we will define a strategy $\mystrat$ that admits the following, key result:

\begin{thm}
\label{thm:correctness}
Let $\psi\in\LL$ be well-formed with AND-OR-graph $\ANDOR$ with initial state $s_0$.
\begin{inparaenum}[(i)]
\item
If $\Win_m(s_0)$ is satisfiable then $\psi$ is boundedly realizable by the strategy $\mystrat$.
\item
If $\psi$ is boundedly realizable with bound $m$ then $\Win_m(s_0)$ is satisfiable.
\item
If the fixpoint $\Win(\ANDOR)$ is defined but unsatisfiable then $\psi$ is not boundedly realizable.
\end{inparaenum}
\end{thm}

Next we define how the formula $\Win(\ANDOR)$, if satisfiable, gives rise to a strategy $g$ to synthesize $\psi$.
For an assignment $\beta$ with domain $X$, let $\pre \beta$ be the assignment with domain $\pre X$ s.t. $\pre \beta(\pre x) = \beta(x)$ for all $x\in X$, and for assignments $\alpha$, $\gamma$ with disjoint domains we write $\alpha \cup \gamma$ for the combined assignment with domain $\dom(\alpha) \cup \dom(\gamma)$.

\begin{definition}[Strategy $\mystrat$]
\label{def:strategy}
Let $\psi$ be well-formed and $\ANDOR$ its AND-OR-graph with initial node $s_0$ such that $\Win_K(s_0)=\Win_{K+1}(s_0)$ for some $K$ and this formula is satisfiable by a $\TT$-model $M$.
For $\pi =\beta_0, \beta_1, \beta_2, \dots\in \assign{X}^*$ the environment assignments, we define some $\len$,
$0 \leq \len\leq K$,
AND-nodes $s_k$ of $\ANDOR$ %for all $0 \leq k \leq \len(\pi)$,
and
assignments $\mystrat(\pi|_{k})$ with domain $Y$ for all $0 < k \leq \len$, by induction on $k$:
\footnote{We assume $\psi \neq \top$, otherwise $\len{=}0$.}
% \begin{compactenum}[(1)]
% \item
% If $k{=}0$ then $s_0$ is the initial node of $\ANDOR$.
% If $s_{0}$ is final then $\len(\pi){=}0$.
% \item
% Suppose $k > 0$.
% First, $s_0$ is the initial node of $\ANDOR$.

Let $\eta_k$ be the assignment combining values from the previous instant and the environment assignment, defined as:
\begin{inparaenum}[(i)]
\item
If $k\,{=}\,0$ then $\dom(\eta_0)\,{=}\,X$ and
$\eta_0=\beta_0$.
\item
If $k > 0$ then $\dom(\eta_k) = X\cup \pre X \cup \pre Y$ and
$\eta_k = \beta_k \cup \pre \beta_{k-1} \cup \pre\gamma_{k-1}$, where $\gamma_{k-1}=\mystrat(\pi|_{k})$.
\end{inparaenum}

Let $s_{k} \goto{w_i} t_i$ be the unique edge from $s_{k}$ % in $\ANDOR$,
such that $\eta_k \models w_i$.
For $t_i \goto{\varsigma_j} u_{j}$, $1 \leq j \leq n_i$ all edges from $t_i$,
let $j$ be such that the formula with free variables $Y$:
\begin{equation}
\label{eq:strategy}
\eta_k(\varsigma_j \wedge \Win_{K-k}(u_{j}))
\end{equation}
is satisfiable in $M$. Let $\mystrat(\pi|_{k+1})$ be a satisfying assignment for it, and set the node $s_{k+1}$ to $u_j$.
If $s_{k+1}$ is final then $\len=k+1$.
% \end{compactenum}
\end{definition}

% We can show that under the assumptions of \defref{strategy}, the strategy $g$ is well-defined.

\noindent
% It can be shown that $g$ is well-defined:

\begin{restatable}{lemma}{lemmag}
The strategy $g$ is well-defined.
\end{restatable}
\begin{proof}[Proof (sketch)]
For all $k$ a unique $w_i$ exists since by construction of $\ANDOR$
the set of outgoing edges from AND-nodes are covering and disjoint.
Second, satisfiability of formula
\eqref{eq:strategy} in $M$ can be shown by induction on $k$, using the fact that $\Win_{K-k}(s_k)$ must be defined.
Finally, the fixpoint $\Win(\ANDOR)$ exists so there must be some $m\leq K$ s.t. $s_{m}$ is labelled $\top$, so $\len$ is defined.
\end{proof}

%%%%%%%%%%%%%%%%%%%%%%%%%%%%%%

The proof of \thmref{correctness} uses the following technical lemma, which relates trace prefixes leading to an AND-node $s$ in $\ANDOR$ to correctness of the strategy $\mystrat$ ``starting'' from that node $s$.
For $\alphas=\tup{\alpha_0,\dots,\alpha_{k-1}}$, we denote by $\alphas|_X$ the sequence $\tup{\alpha_0|_X,\dots,\alpha_{k-1}|_X}$

\begin{lemma}
\label{lem:toll}
Let $\ANDOR[\psi] = (S_\land, S_\lor, E, s_0, \lab)$ be an AND-OR-graph for a well-formed $\psi$, $s_0 \gotos{\rho} s$ for some $s\in S_\land$, and $\tau=(M,\alphas)$ a trace consistent with $\rho$, for $\alphas=\tup{\alpha_0,\dots,\alpha_{k-1}}$.
Let $\alpha_{pre}$ be $\alpha_{k-1}$ if $k>0$ and $\emptyset$ otherwise.
\begin{compactenum}
\item
If $M, \alpha_{pre} \modelsT \Win_m(s)$ for some $m$ then for all
$\pi=\beta_{k},\beta_{k+1},\beta_{k+2}, \dots$ there is some $i \leq m$  such that
$\tau' = \langle\beta_{k}\cup\gamma_{k}, ..., \beta_{k+i} \cup \gamma_{k+i}\rangle$ satisfies $\tau\tau' \models \psi$,
where $\gamma_{k+j}=g(\alphas|_X\cdot\pi|_{j+1})$ for all $0 \leq j \leq i$.
\item
Suppose there are some strategy $h\colon \assign{X}^* \to \assign{Y}$ and $m \geq 0$ such that
for all $\pi=\beta_{k},\beta_{k+1},\beta_{k+2}, \dots$ there is some $i \leq m$  such that
$\tau' = (M,\langle\beta_{k}\cup\gamma_{k}, ..., \beta_{k+i} \cup \gamma_{k+i}\rangle)$ satisfies $\tau\tau' \models \psi$,
where $\gamma_{k+j}=h(\alphas|_X,\beta_{k},...,\beta_{k+j})$ for all $0 \leq j \leq i$. Then $M, \alpha_{pre}\modelsT \Win_m(s)$.
\end{compactenum}
%
% Then $M, \widehat \alpha \modelsT \Win_m(s)$ iff
% % NOTE TO MYSELF:
% % the following formulation is too lax as respective gamma exist also for unrealizable specs with an oracle that can foresee the future
% % for all $\pi = \beta_{k+1}, \beta_{k+2}, \dots$ there are some $i \leq m$ and assignments $\gamma_{k+1}, \gamma_{k+2}, \dots ,\gamma_{k+i}$ such that $\tau' = \langle \beta_{k+1} \cup \gamma_{k+1},  \dots ,\beta_{k+i} \cup \gamma_{k+i}\rangle$ satisfies $\tau\tau' \models \psi$.
% there is some strategy $h\colon \assign{X}^* \to \assign{Y}$ which induces
% for all $\pi=\beta_{k+1},\beta_{k+2},\beta_{k+3}, \dots$  some $i \leq m$  such that
% $\tau' = \langle\beta_{k+1}\cup\gamma_{k+1}, ..., \beta_{k+i} \cup \gamma_{k+i}\rangle$ satisfies $\tau\tau' \models \psi$,
% where $\gamma_{k+j}=h(\alpha_0|_X, ...,\alpha_{k-1}|_X,\beta_k+1,...,\beta_k+j)$ for all $1 \leq j \leq i$.
\end{lemma}
\begin{proof}
($\Longrightarrow$)
By induction on $m$.
If $m=0$ then $s$ must be final.
Let $\sigma=\seq_\psi(\tau)$.
By Lem.~\ref{lem:andor} (1), $\psi^+(\sigma) = \top$.
As $\tau$ is consistent with $\rho$, $\sigma$ satisfies $\rho$, so by \thmref{progression:satisfaction} $\tau \models \psi$, and the claim holds for $\tau'$ empty.

Let $m > 0$.
Since $M,\alpha_{pre} \models \Win_{m}(s)$, either $M,\alpha_{pre} \models \Win_{m-1}(s)$ or $M,\alpha_{pre} \models \PreC_{\Win_{m-1}}(s)$. In the former case, we conclude by the induction hypothesis, so suppose the latter.
% If $s$ is final then we can reason as in the base case, so assume that $s$ is not final and labelled $\psi'$.
Let $s$ be labeled $\psi'$.
There must be exactly one $w_i$ and $t_i$ with $s \goto{w_i} t_i$ such that $M,\eta_{k} \models w_i$, where $\eta_k$ is defined as in \defref{strategy}.
% If $k=0$, $w_i$ contains only variables $X$ by well-formedness, and $M,\beta_{k+1} \models w_i$, or
By definition of $\PreC$, it must holds that
$\eta_k \models \PreC'_{\Win_{m-1}}(t_i)$, so
there must be some $j$ such that
$\eta_k(\varsigma_j \wedge \Win_{m-1}(s_j))$
is satisfiable in $M$ by some $\gamma_{k} \in \assign{Y}$, as assigned by $g(\alphas|_X\cdot\pi|_{k+1})$.
Let $\widehat \tau = (M,\langle \beta_{k} \cup \gamma_{k}\rangle)$.
We have that $\tau\widehat \tau$ is consistent with $\rho$ extended with $s \goto{w_i}t_i \goto{\varsigma_j} s_j$, and we have $M,\beta_{k} \cup \gamma_{k} \models \Win_{m-1}(s_j)$.
By induction hypothesis, there are some $i < m$ and assignments $\gamma_{k+1}, \dots ,\gamma_{k+i}$ produced by $g$ such that for $\tau' = (M,\langle \beta_{k+1} \cup\gamma_{k+1}, \dots ,\beta_{k+i} \cup \gamma_{k+i}\rangle)$ we have $\tau\widehat \tau\tau' \models \psi$.
So the trace $\widehat \tau\tau'$ of length $i+1 \leq m$ satisfies the claim.
\smallskip

($\Longleftarrow$)
By induction on $m$.
If $m=0$ then $\tau \models\psi$.
By \thmref{progression:satisfaction}, $\sigma{=}\seq_\psi(\tau)$ satisfies $\psi^+(\sigma){=}\top$.
By \lemref{andor} (2), there are a unique $\rho'$ satisfied by $\sigma$ such that $\smash{s_0 \gotos{\rho'} s_\ell}$, and $s_\ell$ is final. So it must be $\rho = \rho'$ and $s=s_\ell$,
% Since $\tau$ is consistent with $\rho$, $\sigma$ satisfies $\rho$, so $s=s_\ell$ must be final and
thus $\Win_0(s)=\top$.

Let $m > 0$.
% If $s$ is final, $\Win_{m}(s)=\top$ and the claim holds.
% So suppose $s$ is not final.
If for all $\pi$ there exists an $i<m$ that satisfies the claim, we can conclude that $M,\alpha_{pre} \models \Win_{m-1}(s)$ by the induction hypothesis, and hence $M,\alpha_{pre} \models \Win_{m}(s)$ by definition of $\Win_{m}$.
Otherwise, we show that $M,\alpha_{pre} \models \PreC_{\Win_{m-1}}(s)$, by showing that
% $\widehat\alpha(\PreC_{\Win_{m-1}}(s))$ is satisfiable, which is the case iff each conjunct
% $\widehat\alpha_{pre}(\forall X.\: w_i \to  \PreC'_{\Win_{m-1}}(t_i))$
$M,\pre\alpha_{pre}\models \forall X.\: w_i \to  \PreC'_{\Win_{m-1}}(t_i)$
for all $s \goto{w_i} t_i$.
To this end, we show that
$M,\pre\alpha_{pre}\models \widehat\beta(\: w_i \to  \PreC'_{\Win_{m-1}}(t_i))$ (called ($\star$) below),
where $\widehat \beta$ is an arbitrary assignment with domain $X$.
First, if $\widehat\beta(w_i)$ is unsatisfiable in $M$, then ($\star$) is trivially satisfied.
Otherwise, consider some $\pi = \beta_{k}, \beta_{k+1}, \dots$ such that $\beta_{k}=\widehat\beta$. Then there are some $i \leq m$ and assignments $\gamma_{k}, \gamma_{k+1}, \dots ,\gamma_{k+m}$ such that $\tau' = \langle \beta_{k} \cup \gamma_{k}, \dots ,\beta_{k+i} \cup \gamma_{k+i}\rangle$ satisfies $\tau\tau' \models \psi$, where $\gamma_{k}, \dots \gamma_{k+m}$ are produced by strategy $h$.
Since the edge labels $\varsigma_1, \dots, \varsigma_n$ such that there is a transition $t_i \goto{\varsigma_j} s_{j}$, $1\leq j \leq n$ are covering and distinct, there must be a unique $\varsigma_j$ such that
$M,\pre \alpha_{pre} \cup \widehat\beta\cup \gamma_{k}\models \varsigma_j$.
Now for the trace $\tau\langle\beta_{k}\cup \gamma_{k}\rangle$ that is compatible with $\rho' = \rho\langle w_i \wedge \varsigma_j\rangle$ it holds that $s_0 \goto{\rho'} s_j$. Moreover, by assumption, for all $\beta_{k+1}, \beta_{k+2},\dots$, it holds that this sequence can be extended to a trace $\tau''$ of length $i<m$ such that $M,\tau\langle\beta_{k}\cup \gamma_{k}\rangle\tau'' \models \psi$.
Thus, by induction hypothesis, $M,\beta_{k}\cup \gamma_{k} \models \Win_{m-1}(s_j)$,
so that $\pre\alpha_{pre}$ satisfies ($\star$).
We conclude that $M,\alpha_{pre} \models \PreC_{\Win_{m-1}}(s)$, hence $M,\alpha_{pre} \models \Win_{m}(s)$.
\end{proof}

%%%%%%%%%%%%%

\noindent
At this point it is easy to prove \thmref{correctness}:
\begin{proof}
Items $(i)$ and $(ii)$ follow from Lem.~\ref{lem:toll}~(1) and (2) in the case of an empty trace $\tau$ and $s=s_0$.
Item $(iii)$ follows from $(ii)$.
\end{proof}

% From Ex.~\ref{exa:preimage} we can thus conclude that the property in Exa.~\ref{exa:andor} is realizable with bound 2.

\begin{example}
Assume that for the formula in Ex.~\ref{exa:andor},
the environment chooses assignments $\pi = \langle \{x\mapsto 3\},\{x \mapsto 4\}, \dots\rangle$.
\begin{compactenum}
\item
For $\pi|_1$, the corresponding node $s_1$ is labelled $\psitwo$.
By \exaref{preimage}, $\Win_1(s_1) = y>x+2 \vee x < 0$, so for $x=3$ one needs to choose a value of $y > 5$, e.g. $g(\pi|_1) = \{y \mapsto 6\}$.
\item
For $\pi|_2$, we have $s_2 = \top$, so $g(\pi|_2)$ is arbitrary.
\end{compactenum}
\end{example}

Note that for the case of propositional \LTLf, realizability is implicitly bounded by the number of states in the DFA. The same holds for the first-order case without lookback, which can be reduced to a boolean abstraction.
However, the following example shows that in the case of \LTLf modulo theories, realizability need not be bounded.

\begin{example}
\label{exa:unbounded}
Let $x$ and $u$ be controlled by the environment, and $y$ by the agent.
Consider the assumption $A = (u \geq 0) \wedge (x=0) \wedge \X \G(x = \pre x  + 1)$ and the property
$\psi = A \to (y=u) \wedge \F (x=y) \wedge \X\G (y=\pre y)$.
In words, the assumption is that $u$ is initially non-negative, $x$ is initially 0 and increases by 1 in every step;
and it is required that $y$ maintains the value of the environment-controlled $u$  throughout the trace, and has at some point the value of $x$.
This property is clearly realizable by the strategy that sets $y$ initially to the value of $u$ and keeps this value.
However, the length of the strategy is the initial value of $u$, which is arbitrarily large.
\end{example}

\section{Decidability}
\label{sec:decidability}

In this section, we identify fragments of \LTLfMT where the synthesis problem for $\psi$ is \emph{ solvable}, in the sense that
the realizability problem is decidable and the strategy $g$ from \defref{strategy} realizes $\psi$.

Precisely, we obtain decidability for the following four fragments:
\begin{inparaenum}[(1)]
\item Lookback-free properties over $\TT$ with decidable $\forall^*\exists^*$ fragment,
\item monotonicity constraints, i.e., variable-to-variable/constant comparisons over \LRA,
\item integer periodicity constraints, and 
\item bounded lookback properties, if satisfiability of formulas with a sufficient number of quantifier alternations are decidable in $\TT$ .
\end{inparaenum}
\medskip

We start with \textbf{lookback-free properties}, for which we obtain a similar decidability result as in~\cite[Thm. 2]{RodriguezS23}:

\begin{restatable}{thm}{thmnolookback}
\label{thm:no:lookback}
If satisfiability of $\forall^*\exists^*$ formulas is decidable in $\TT$,
the synthesis problem is solvable for lookback-free properties.
\end{restatable}
% \begin{proof}[Proof (sketch)]
% A simple inductive argument shows that $\Win_i(s)$  is either $\top$ or $\bot$ for all $i \geq 0$, where $\PreC_{\Win_i}$ can be computed by checking $\TT$-satisfiability of a $\forall^*\exists^*$-formula, or validity of a $\exists^*\forall^*$-formula.
% The fixpoint computation must hence terminate.
% \end{proof}

The proof can be found in the appendix, it essentially shows that $\Win_i(s)$ can be expressed as a $\forall^*\exists^*$ sentence for all $i \geq 0$, and due to the absence of lookback, the fixpoint exists. Satisfiability of $\Win(\ANDOR)$ can thus be computed by checking $\TT$-satisfiability of a $\forall^*\exists^*$-formula (or equivalently, validity of an $\exists^*\forall^*$-formula).

% This generalizes the decidability result \cite[Thm.~2]{RodriguezS23}, where $\TT$-satisfiability of $\exists^*\forall^*$ formulas must be decidable.
Several theories $\TT$ fulfill the decidability assumption on $\forall^*\exists^*$ formulas in \thmref{no:lookback}, e.g. linear arithmetic and the Bernays-Sch\"onfinkel class~\cite{BS:MA:1928}.
Reactive synthesis of properties without lookback has e.g. been applied to shield synthesis  of neural networks~\cite{rodriguez25shield}.

We next focus on linear arithmetic theories:

\paragraph{Monotonicity constraints} (MCs) have been first considered for counter systems~\cite{DD07}. %, and are also known as monitonicity constraints (MCs)~\cite{FMPW23}.
In this case, all atoms in $\LL$ are variable-to-variable and variable-to-constant comparisons, i.e., of the form $p \odot q$ where $p,q\in {\mathbb Q\,{\cup}\,V}$
and $\odot$ is one of $=, \neq, \leq$, or $<$.
We call a property in $\LL$ over the theory \LRA where all atoms have this shape an \emph{MC property}.
For instance, $(x\,{>}\,\pre x) \U (x{=}y \wedge y{=}10)$ is an MC property, but the property from \exaref{intro} is not.
The following result is shown in a similar way as \cite[Thm.~33]{FMPW23}, exploiting the fact that quantifier elimination of MC formulas over $\mathbb Q$ produces MC formulas with the same constants. The proof can be found in the appendix.

\begin{restatable}{thm}{theoremMC}
\label{thm:mc}
The synthesis problem is solvable for MC properties.
\end{restatable}

For \LTL over infinite traces, a respective result is shown in \cite[Thm. 10]{BhaskarP24}.
MC properties are expressive enough to capture practical applications:
E.g., the specification of a fluxgate magnetometer of a Swift UAS system run by NASA can be modelled via \LTLfMT properties over MCs~\cite[Table 2]{GeistRS14}.
Data Petri nets (DPNs) with MC transition guards are also a popular model in business process management (BPM)
 that can be mined automatically from event logs~\cite{LM18}.
In BPM, processes typically involve multiple, possibly non-cooperating parties, so that reactive synthesis is highly relevant in this context.
% 
% \begin{proof}[Proof (sketch)]
% This is shown as in~\cite[Thm.~5.2]{FMPW23}, exploiting that quantifier elimination of MC formulas over $\mathbb Q$ produces MC formulas with the same  constants, so that the history formulas in $\ANDOR$ range over a finite set.
% \end{proof}

\paragraph{Integer periodicity constraints} (IPCs) are atoms in the theory \LIA that somewhat resemble MCs, but allow equality modulo while variable-to-variable comparisons with $<$ and $\leq$ are excluded~\cite{Demri06}.
Precisely, IPC atoms have the form $x = y$, $x \odot d$ for $\odot \in \{=,\neq, <, >\}$, $x \equiv_k y + d$, or $x \equiv_k d$, for variables $x,y$ with domain $\mathbb Z$ and $k,d\in \mathbb N$.
We call $\psi \in \LL$ an \emph{IPC property} if all atoms in $\psi$ are IPCs.
For instance, $\F\G(x\,{=}\,\pre x) \wedge \G(y \equiv_3 0)$ is an IPC property, but $\G(x>y)$ is not.
The next result is proven similarly as Thm.~\ref{thm:mc}, as formulas over IPCs are closed under quantifier elimination~\cite{Demri06}.

\begin{restatable}{thm}{theoremIPC}
\label{thm:ipc}
The synthesis problem is solvable for IPC properties.
\end{restatable}

\paragraph{Bounded lookback} is a theory-independent restriction of the interaction between variables via lookback. Intuitively, a property $\psi \in \LL$ has $K$-bounded lookback, for $K\geq 0$, if there is no path in the AND-OR-graph $\ANDOR$ whose edge labels create a dependency chain between variables that spans more than $K$ instants; where, however, equality atoms are not counted for dependencies as they can be collapsed.
Bounded lookback has been considered for model checking and monitoring~\cite{FMPW23,GMW24}. In fact the definition below differs slightly from the literature, but we keep the name as the concept is the same.

Given a property $\psi \in \LL$ and $\ANDOR[\psi]$ with AND-nodes $s$, $s'$,
let $\rho=\langle w_0 \wedge \varsigma_0, \dots, w_{n-1} \wedge\varsigma_{n-1}\rangle$ be a sequence of formulas as in Eq.~\eqref{eq:path} such that $s \gotos{\rho} s'$.
For each $0\,{\leq}\,i\,{<}\,n$, let $V_i$ be a fresh set of variables of the same size and sorts as $V$, and $\mathcal V = \bigcup_{i=0}^{n-1} V_i$.
Consider the formulas
$(w_0 \wedge \varsigma_0)[V_0]$, i.e., the formula obtained from $w_0 \wedge \varsigma_0$ by systematically replacing $V$ with $V_0$; and $(w_i \wedge \varsigma_i)[V_{i-1}, V_i]$, i.e., the formula obtained from $w_i \wedge \varsigma_i$ by systematically replacing $\pre V$ with $V_{i-1}$ and $V$ with $V_i$, for $0<i<n$.
The \emph{computation graph} $G_\rho$ of $\rho$ is the undirected graph with nodes $\mathcal V$ and an edge from $u\in \mathcal V$ to $v\in \mathcal V$ iff the two variables occur in a common literal of
$(w_0 \wedge \varsigma_0)[V_0] \wedge \bigwedge_{i=1}^n(w_i \wedge \varsigma_i)[V_{i-1}, V_i]$.
% The subgraph of $G_{w}$ of all edges corresponding to equality literals $x=y$ for $x, y \in \mc V$ is denoted $E_{w}$.
Finally, $[G_{\rho}]$ denotes the graph obtained from $G_{\rho}$ by collapsing
all edges corresponding to equality literals $u=v$ for $u, v \in \mc V$.

\begin{definition}
$\ANDOR[\psi]$ has \emph{$K$-bounded lookback} if for all AND-nodes $s$, $s'$ and
all sequences of formulas $\rho$ such that $s \gotos{\rho} s'$ in $\ANDOR$,
all acyclic paths in $[G_{\rho}]$ have length at most $K$.
%
% $\ANDOR[\psi]$ has \emph{bounded lookback} if it has $K$-bounded lookback for some $K$.
\end{definition}

E.g., $(x=y) \U (x >\pre x \wedge \X(x{>}\pre x))$ has 3-bounded lookback, but $\psi$ from \exaref{intro} or $(x>\pre x) \U (x=y)$ do not have $K$-bounded lookback for any $K$.
We show next that the synthesis problem for bounded lookback properties can be solved if FO formulas with a sufficient number of quantifier alternations can be decided in $\TT$.
The proof basically shows that every formula $\Win_i(s)$ comes from a set of formulas that is finite because it has bounded quantifier depth and a finite vocabulary; for details cf. the appendix.

\begin{restatable}{thm}{thmBL}
\label{thm:bl}
If the theory $\TT$ has a decidable $(\forall^*\exists^*)^K$ fragment,
the synthesis problem is solvable for $K$-bounded-lookback properties.
\end{restatable}

More concretely, \thmref{bl} applies to all $K$ if $\TT$ has quantifier elimination, as is the case for \LIA or \LRA, or if $\TT$ is a decidable first order class such as FO$^2$~\cite{Mortimer75} or the stratified fragment~\cite{ABN:JPL:1998}.
Rich models to describe business processes such as DPNs 
and data-aware processes modulo theories~\cite{GMW24}
can be encoded in \LTLfMT properties over linear arithmetic, and possibly uninterpreted functions.
In fact many practical business processes were shown to enjoy bounded lookback, or the weaker feedback freedom  property~\cite{GMW24,FMW22caise,DDV12}.

\section{Conclusion}
\label{sec:conclusion}

This paper presents the first reactive synthesis procedure for \LTLfMT properties with lookback.
The approach is proven sound, and complete wrt. \emph{bounded realizability}.
It is also a decision procedure for several fragments of \LTLfMT, reproving known decidability results and identifying new decidable fragments.
% including properties where variable interaction via lookback is bounded, with a reasonable restriction on decidability of the underlying logic; and properties over the theory of linear arithmetic where atoms are variable-to-variable/constant comparisons. These outcomes cover at once results from several earlier works~\cite{RodriguezS23,LFM20,BhaskarP24}.

A variety of directions for future work arise. First of all, we plan an implementation with a focus on linear arithmetic, using SMT solvers with QE support as backends~\cite{Z3,cvc5}. Though the treatment of quantifiers is notoriously difficult, recent verification tools witness that QE is not a showstopper~\cite{HiplerKLS24,GMW24,FMPW23}. Second, efficiency could be gained by interleaving the computation of the AND-OR-graph $\ANDOR$ and the winning condition $\Win(\ANDOR)$, exploiting on-the-fly techniques from~\cite{GiacomoFLVX022,XiaoL0SPV21}. Third, it would be interesting to study whether the restriction on bounded realizability in \thmref{correctness} can be lifted.

\bibliography{references}
\newpage
\appendix
\section{Proofs}

\begin{lemma}
\label{lem:boring}
Let $\psi$ be well-formed and $\sigma = \seq_\psi(\tau)$. Then $\sigma \models \psi$ in a propositional sense iff $\tau \models \psi$.
\end{lemma}
\begin{proof}
Let $\tau = (M,\tup{\alpha_0, \dots, \alpha_{n-1}})$ and $\sigma=\tup{A_0, \dots, A_{n-1}}$.
We first show that for all $i$, $0 \leq i < n$, and $\chi$ a subproperty of $\psi$ such that either $i>0$ or $\chi$ is well-formed, $\tau,i \models \chi$ iff $\sigma,i \models \chi$, by induction on $\chi$.
\begin{compactitem}
\item If $\chi$ is an atom, $\tau,i \models \chi$ implies $M,\alpha_0 \models \chi$ iff $i=0$, or $M,\combine{\alpha_{i-1}}{\alpha_i} \models \chi$ otherwise. By definition of $\seq_\psi$, this is the case iff $\psi \in A_i$ (note that since $\chi$ is well-formed if $i=0$, $\chi$ does not contain $\pre V$ in this case).
\item If $\chi = \neg a$ for an atom $a$, or $\chi_1 \wedge \chi_2$ or $\chi_1 \vee \chi_2$, we conclude by the induction hypothesis. Note that in all cases well-formedness of $\chi$ implies well-formedness of $a$ resp. $\chi_1$ and $\chi_2$.
\item If $\chi =\X \chi'$ then $\tau,i+1 \models \chi'$ iff $\sigma,i+1 \models \chi'$ by the induction hypothesis.
We have $\tau,i \models \chi$ iff $i < n-1$ and $\tau,i+1 \models \chi'$, which holds  iff $i < n-1$ and $\sigma,i+1 \models \chi'$, which holds iff $\sigma,i \models \chi$.
\item The case of $\chi=\Xw \chi'$ is similar.
\item 
The cases of $\chi_1 \U \chi_2$ and $\chi_1 \R \chi_2$ combine reasoning from the cases above.
\end{compactitem}
We thus conclude that $\tau,0 \models \psi$ iff $\sigma,0 \models \psi$.
\end{proof}

Let $\sim$ denote propositional equivalence of properties in $\LL$, i.e., the equivalence relation obtained by applying boolean equivalences.
Progression of $\sim$-equivalent properties yields $\sim$-equivalent properties~\cite[Lem. 2]{GiacomoFLVX022}, below we will use this fact ($\star$).

\thmbasic*
\begin{proof}
By \lemref{progression}, $\psi^+(\sigma)=\top$ iff $\sigma \models \psi$ in a propositional sense, and by \lemref{boring}, $\sigma \models \psi$ in a propositional sense iff $\tau \models \psi$.
\end{proof}

\begin{lemma}
\label{lem:xnf:decomposition}
Let $\psi$ be in XNF and $\decomp(\psi, C)=\{(\prim_1,\sub_1), \dots,(\prim_k,\sub_k)\}$ for some set of atoms $C$ such that $\foa(\psi) \subseteq C$ .
If $A$ is a set of atoms such that $A\models \prim_i$ for some $1\leq i \leq k$
then $\psi^+(A)=\remX(\sub_i)$.
\end{lemma}
\begin{proof}
By induction on $\psi$.
\begin{compactitem}
\item
If $\psi=\top$ then the only valid decomposition is $\decomp(\psi, C)=\{(\top,\top)\}$. We have $\psi^+(A)=\top = \remX(\top)$.
\item
If $\psi=\bot$ then the only valid decomposition is $\decomp(\psi, C)=\{(\top,\bot)\}$. We have $\psi^+(A)=\bot = \remX(\bot)$.
\item
Suppose $\psi$ is an atom. Then the only valid decomposition is $\decomp(\psi, C)=\{(\psi,\top), (\neg\psi, \bot)\}$.
If $\psi \in A$ then $\psi^+(A)=\top = \remX(\top)$.
If $\psi \not\in A$ then $\psi^+(A)=\bot = \remX(\bot)$.
\item
If $\psi = \X \psi'$ we must have $\decomp(\psi, C)=\{(\top,\psi)\}$.
We have $\psi^+(A)=\psi' \wedge \neg \last = \remX(\psi)$.
\item
If $\psi = \Xw \psi'$ we must have $\decomp(\psi, C)=\{(\top,\psi)\}$.
We have $\psi^+(A)=\psi' \vee \last = \remX(\psi)$.
\item
Let $\psi = \psi_1\wedge \psi_2$. Both $\psi_1$, $\psi_2$ are in XNF.
Let $\decomp(\psi_1, \foa(\psi_1))=\{(p_1,s_1), \dots,(p_k,s_k)\}$, and 
$\decomp(\psi_2, \foa(\psi_2))=\{(q_1,r_1), \dots,(q_m,r_m)\}$.
Since decompositions are covering and disjoint, there must be exactly one $I$ such that $A\models p_I$ and exactly one $J$ such that $A\models q_J$.
By the induction hypothesis, $\psi_1^+(A)=\remX(s_I)$ and $\psi_2^+(A)=\remX(r_J)$.
By \cite[Thms. 2 and 3]{Darwiche11}, the only compressed decomposition for $\psi$ is 
\[\{(p_l\wedge q_j, s_l \wedge r_j) \mid 1\leq l \leq k\text{, }1\leq j \leq m\text{, and }p_i\wedge q_j \not\equiv\bot\}\]
So we must have $\prim_i \equiv p_I \wedge q_J$ and $\sub_i \equiv s_I \wedge r_J$.
We have $\psi^+(A) = \psi_1^+(A) \wedge \psi_2^+(A) = \remX(s_I) \wedge \remX(r_J) = \remX(s_I \wedge r_J) \equiv \remX(\sub_i)$, using ($\star$).
\item
Let $\psi = \psi_1\vee \psi_2$, where both $\psi_1$, $\psi_2$ are in XNF.
Let $\decomp(\psi_1, \foa(\psi_1))=\{(p_1,s_1), \dots,(p_k,s_k)\}$, and
$\decomp(\psi_2, \foa(\psi_2))=\{(q_1,r_1), \dots,(q_m,r_m)\}$.
Since decompositions are covering and disjoint, there must be exactly one $I$ such that $A\models p_I$ and exactly one $J$ such that $A\models q_J$.
By the induction hypothesis, $\psi_1^+(A)=\remX(s_I)$ and $\psi_2^+(A)=\remX(r_J)$.
By \cite[Thms. 2 and 3]{Darwiche11}, the only compressed decomposition for $\psi$ is
\[\{(p_l\wedge q_j, s_l \vee r_j) \mid 1\leq l \leq k\text{, }1\leq j \leq m\text{, and }p_i\wedge q_j \not\equiv\bot\}\]
So we must have $\prim_i \equiv p_I \wedge q_J$ and $\sub_i \equiv s_I \vee r_J$.
We have $\psi^+(A) = \psi_1^+(A) \vee \psi_2^+(A) = \remX(s_I) \vee \remX(r_J) = \remX(s_I \vee r_J) = \remX(\sub_i)$.
\end{compactitem}
No other operators are relevant for a property in XNF.
\end{proof}

\lemmaandor*
\begin{proof}
\renewcommand{\prim}{\mathsf{p}}
\renewcommand{\sub}{\mathsf{s}}
\begin{compactenum}[(1)]
\item
We show the more general statement that if $\rho$ is a sequence of formulas 
such that $s \gotos{\rho} s_\ell$ for some AND-node $s$ labeled $\chi$ and some final node $s_\ell$,
then for every sequence of atom sets $\sigma$ that satisfies $\rho$, it holds that
$\chi^+(\sigma) = \top$; the claim then follows from the case of $s=s_0$ where $\chi=\psi$. 
By induction on the length of $\rho$.
If $\rho$ is empty, then $s=s_\ell$ which is labeled $\top$, also $\sigma$ must be empty, and $\top^+(\epsilon) = \top$.
So let $\rho = \langle w \wedge \varsigma\rangle\cdot \rho'$ and
$s \goto{w} u \goto{\varsigma}{s'} \gotos{\rho} s_\ell$ such that $s' = (\chi', \phi')$.
For any $\sigma$ that satisfies $\rho$, we can write $\sigma = \langle A\rangle\cdot\sigma'$ such that $\sigma'$ satisfies $\rho'$.
By the induction hypothesis, $\chi'^+(\sigma')=\top$.
We can assume that $\chi$ is in XNF.
Let $\decomp(\chi, C_\env)=\{(\prim_1,\sub_1), \dots,(\prim_k,\sub_k)\}$.
By construction there must be some $i$, $1\leq i \leq k$, such that $\prim_i=w$, and
for $\decomp(\xnf(\sub_i), C_\ag)=\{(\prim_{i,1},\sub_{i,1}),\dots,(\prim_{i,m},\sub_{i,m})\}$, there is some $j$, $1\leq j \leq m$, such that $\varsigma=\prim_{i,j}$.
Moreover, we must have $\chi' = \xnf(\remX(\sub_{i,j}))$.
In addition, from the properties of a decomposition it follows that the following is a valid decomposition wrt. $C = \foa(\chi)$:
\begin{align*}
\decomp(\chi, C)= \{&(\prim_1,\sub_1), \dots,(\prim_{i-1},\sub_{i-1}),\\
&(\prim_i \wedge \prim_{i,1},\sub_{i,1}),\dots,(\prim_i \wedge \prim_{i,m},\sub_{i,m}), \\
&(\prim_{i+1},\sub_{i+1}), \dots,(\prim_k,\sub_k)
\}
\end{align*}
We have $A\models \prim_i \wedge\prim_{i,j}$ because $\sigma$ satisfies $\rho$. By \lemref{xnf:decomposition},
$\chi^+(A) = \remX(\sub_{i,j})$.
By the definition of progression for sequences using $\chi'^+(\sigma')=\top$, it follows that
$\chi^+(\sigma) = \top$.
\item
We show the more general statement that if there is an AND-node $s$ labeled $\chi$ and $\chi^+(\sigma) = \top$ for some sequence of atom sets
$\sigma$ over $\foa(\chi)$
then there is a sequence of formulas $\rho$
such that $\sigma$ satisfies $\rho$ and
$s \gotos{\rho} s_\ell$ for $s_\ell$ final; the claim then follows for $\chi=\psi$.
The proof is by induction on the length of $\sigma$.
If $\sigma$ is empty, then $\chi=\top$ is already final.
So let $\sigma = \langle A\rangle\sigma'$, so that $\chi^+(\sigma) = (\chi^+(A))^+(\sigma')$.
We can assume that $\chi$ is in XNF.
Let $\decomp(\chi, C_\env)=\{(\prim_1,\sub_1), \dots,(\prim_k,\sub_k)\}$.
Since the decomposition is covering, there must be some $i$, $1\leq i \leq k$, such that $A\models \prim_i$, let $w:= \prim_i$.
Similarly, for $\decomp(\xnf(\sub_i), C_\ag)=\{(\prim_{i,1},\sub_{i,1}),\dots,(\prim_{i,m},\sub_{i,m})\}$, there is some $j$, $1\leq j \leq m$, such that $A\models\prim_{i,j}$, let $\varsigma := \prim_{i,j}$.
There must be a path $s_0 \goto{w}\cdot\goto{\varsigma} s'$, $s'$ has label $\chi':=\xnf(\remX(\sub_{i,j}))$, and $\ANDOR[\chi]$ contains the AND-OR-graph $\ANDOR[\chi']$ as a subgraph.
The combined decomposition wrt. $C = \foa(\chi)$ is given by
\begin{align*}
\decomp(\chi, C)= \{&(\prim_1,\sub_1), \dots,(\prim_{i-1},\sub_{i-1}),\\
&(\prim_i \wedge \prim_{i,1},\sub_{i,1}),\dots,(\prim_i \wedge \prim_{i,m},\sub_{i,m}), \\
&(\prim_{i+1},\sub_{i+1}), \dots,(\prim_k,\sub_k)
\}
\end{align*}
such that $A\models \prim_i \wedge\prim_{i,j}$. By \lemref{xnf:decomposition},
$\chi^+(C) = \remX(\sub_{i,j})=\chi'$.
By the induction hypothesis, $\ANDOR[\chi']$ has a path $s' \goto{\rho'} s_\ell$ for a final node $s_\ell$ such that $\sigma'$ satisfies $\rho'$.
Hence we can set $\rho = \tup{w \wedge \varsigma}\cdot\rho'$, which is satisfied by $\sigma$, and we have  $s_0 \gotos{\rho} s_\ell$.
\end{compactenum}
\end{proof}

\lemmag*

\begin{proof}
First, in Item (2) of Def.~\ref{def:strategy}, for all $k$ a unique $w_i$ exists such that there is an edge $t_{k-1} \goto{w_i} s_{i}$ for some $s_i$ that satisfies
$\beta_k \models w_i$, because by construction of $\ANDOR$, for every AND-node $s$,
the set of outgoing edges are covering and disjoint.
Second, we show satisfiability of the formula 
Eq.~\eqref{eq:strategy}
in Def.~\ref{def:strategy}, by induction on $k$.
Let $\pi = \alpha_1, \alpha_2, \alpha_3, \dots$ be the assignments chosen by the environment, 
% , and $t_0$, $t_k$, $u_k$, $w_k$, and $\varsigma_k$ be as in Eq.~\ref{eq:strategy:path}, for $1\leq k \leq \ell$.
% We show satisfiability of \eqref{eq:strategy}
$s_i$, $0 \leq i < k$, the AND-nodes defined by Def.~\ref{def:strategy}, 
and $t_i$, $0\leq i<k-1$ the unique OR-node between $s_{i}$ and $s_{i+1}$.

If $k=1$ then 
let $t_0$ be such that $s_0 \goto{w} t_0$ and $w$ is satisfied by $\beta_0$, and $t_0\goto{\varsigma_j} u_j$, $1\leq j \leq m$ be all edges from $t_0$.
By assumption, $\Win_M(s_0)$ is valid, and so is $\PreC_{\Win_{M-1}}(s_0)$ and its subformula
\[
\forall X. (w \to  \bigvee_{j=1}^m \exists Y. \varsigma_j \wedge \Win(u_{j}))
\].
Since $\beta_0 \models w$,
$\beta_0(\bigvee_{j=1}^m Y.  \varsigma_j \wedge \Win(u_{j}))$
is valid, and hence there must be some $j$, $1 \leq j \leq m$ such that the formula
$\beta_0(\varsigma_j \wedge \Win_{M-1}(u_{j}))$ with free variables $Y$ is satisfiable,
which equals \eqref{eq:strategy} as $\eta_1=\alpha_1$.

\smallskip
Let $k > 1$, and $\gamma_{k-1} = g(\pi|_{k-1})$.
By the induction hypothesis, $\beta_{k-1}(\Win_{M-k+1}(s_{k-1}))$ is satisfiable and according to Def.~\ref{def:strategy} satisfied by $\gamma_{k-1}$.
Thus $\gamma_{k-1}(\beta_{k-1}(\Win_{M-k+1}(s_{k-1})))$ is valid.
So there must be some $i < M - k +1$ such that $\gamma_{k-1}(\beta_{k-1}(\PreC_{\Win_{i}}(s_{k-1})))$ is valid.
Let $t_{k-1}$ be such that $s_{k-1} \goto{w} t_{k-1}$ and $w$ is satisfied by $\eta_k$,
and let $t_{k-1}\goto{\varsigma_j} u_j$, $1\leq j \leq m$ be all edges from $t_{k-1}$.
Then
\[
\pre \gamma_{k-1}(\pre \beta_{k-1}(
\forall X. (w \to  
\bigvee_{j=1}^m
\exists Y. \varsigma_j \wedge \Win_i(u_j))
))
\]
is valid where we chose the relevant conjunct from $\PreC_{\Win_{i}}$ corresponding to $w$.
Thus so is
\[
\pre \gamma_{k-1}(\pre \beta_{k-1}(
\beta_k(w \to  
\bigvee_{j=1}^m
\exists Y. \varsigma_j \wedge \Win_i(u_j))
))
\]
i.e. the formula obtained from the former when using the values of $\beta_k$ for $X$. Since $\eta_k = \pre \gamma_{k-1} \cup\pre \beta_{k-1}\cup\beta_k$ satisfies $w$ by definition, there must be some $j$, $1\leq j \leq m$, such that 
$\eta_k(\exists Y. \varsigma_j \wedge \Win_i(u_j))$ is valid, i.e.,
$\eta_k(\varsigma_j \wedge \Win_i(u_j))$ is satisfiable, which coincides with Eq.~\eqref{eq:strategy} because $\Win_i(u_j)$ does not contain $\pre X\cup \pre Y$.
% 
% Finally, this shows that for every $\pi$ there is a path in $\ANDOR$: 
% \[
% t_0 \goto{w_0} s_0 \goto{\varsigma_0} t_1 \goto{w_1} s_1 \goto{\varsigma_1}  
% t_2 \goto{w_2} s_2 \dots
% \]
% for $t_i$ as defined in Def.~\ref{def:strategy}.

Finally, note that there must be some $\ell\geq 0$ such that $s_\ell$ is labelled $\top$ as
otherwise the fixpoint $\Win(\ANDOR)$ could not exist.
\end{proof}

\thmnolookback*
\begin{proof}
An inductive argument shows that $\Win_i(s)$ can be expressed as a $\forall^*\exists^*$-formula without free variables, for all $i \geq 0$ and AND-nodes $s$. This clearly holds for $i=0$. In the inductive step, we assume that the claim holds for $\Win_i(s)$.
First, since all labels in $\ANDOR$ are lookback-free, $\Win_{i}(s)$ has no free variables.
When computing $\PreC_{\Win_i}$, let $v_s$ be a fresh boolean variable for every AND-node $s$. Let $\widehat \PreC_{\Win_i}$ be obtained from $\PreC_{\Win_i}$ by replacing every occurrence of $\Win_i(s)$ by $v_s$.
Then we have $\PreC_{\Win_i} \equiv \bigwedge_{s\in S_\land} v_s \leftrightarrow \Win_i(s) \wedge \widehat \PreC_{\Win_i}$ (using that $\Win_{i}(s)$ has no free variables). This is again a $\forall^*\exists^*$-formula.
Moreover, the fixpoint $\Win(\ANDOR)$ must exist because due to the absence of lookback, unrolling loops in $\ANDOR$ more than once does not result in any change.
Hence, satisfiability of $\Win(\ANDOR)$ amounts to satisfiability of a $\forall^*\exists^*$-formula
(or equivalently, validity of a $\exists^*\forall^*$-formula).
Then the claim follows from \thmref{correctness}.
\end{proof}

\theoremMC*
\begin{proof}
\newcommand{\mcK}{MC$_{\mathcal K}$}
Let $\mathcal K$ be the set of all constants in atoms of $\psi$,
and {\mcK} the set of all quantifier-free boolean formulas where all atoms are MCs over variables $V$ and constants in $\mathcal K$, so {\mcK} is finite up to equivalence.
We show by induction on $i$ that for all AND-nodes $s$ of $\ANDOR$, the formula $\Win_i(s)$ is in \mcK. This clearly holds for $i=0$.
Under the assumption that $\Win_i(s) \in \text{\mcK}$ for all $s$, we show that $\PreC_{\Win_i}(s) \in \text{\mcK}$, from which the claim follows for the case of $i+1$.
Indeed, by definition 
\begin{equation}
\label{eq:preCeq}
\PreC_{\Win_i}(s) = \bigwedge_i \forall X. w_i \to \bigvee_j \varsigma_j  \wedge \Win_i(s_{ij})
\end{equation}
where all of $w_i$, $\varsigma_j$, and $\Win_i(s_{ij})$ are in \mcK.
LRA has quantifier elimination, and a quantifier elimination procedure \'{a} la Fourier-Motzkin ~\cite[Sec. 5.4]{KS16}
can produce a quantifier-free formula equivalent to \eqref{eq:preCeq} that is again in \mcK.

Therefore all formulas $\Win_i(s)$ for and AND-node $s$ of $\ANDOR$ and $i\geq 0$ come from a set that is finite up to equivalence, so whenever computing a new formula $\Win_i(s)$, it can be checked whether it is equivalent to an already computed one.
Thus the fixpoint $\Win(\ANDOR)$ must exist, and the synthesis problem can be solved via \thmref{correctness}.
\end{proof}

\theoremIPC*
\begin{proof}
\newcommand{\ipcK}{IPC$_{\mathcal K}$}
Let $\mathcal K$ be the set of all constants in atoms of $\psi$,
and {\ipcK} the set of all quantifier-free boolean formulas where all atoms are MCs over variables $V$ and constants in $K$, so {\ipcK} is finite up to equivalence.
We show by induction on $i$ that for all AND-nodes $s$ of $\ANDOR$, the formula $\Win_i(s)$ is in \ipcK. This clearly holds for $i=0$.
Under the assumption that $\Win_i(s) \in \text{\ipcK}$ for all $s$, we show that $\PreC_{\Win_i}(s) \in \text{\ipcK}$, from which the claim follows for the case of $i+1$.
Indeed, in the definition of $\PreC$, as recalled in \eqref{eq:preCeq},
all of $w_i$, $\varsigma_j$, and $\Win_i(s_{ij})$ are in \ipcK.
By \cite[Thm. 2]{Demri06}, quantifier elimination of a formula in {\ipcK} yields again a formula in \ipcK.

We conclude as in \thmref{mc} that the synthesis problem is solvable.
\end{proof}

\thmBL
\begin{proof}
Let $\psi$ have $K$-bounded lookback, for some $K>0$, and $\Phi$ be the set of formulas with free variables $V$, 
quantifier depth at most $2K$, and using as atoms all atoms in 
$\psi$, but where variables in $\pre V \cup V$ may be replaced by $\mathcal V$.
Being a set of formulas with bounded quantifier depth over a finite set of atoms, $\Phi$ is finite up to equivalence.

We show by induction on $i$ that for all AND-nodes $s$ of $\ANDOR$, the formula $\Win_i(s)$ is equivalent to a formula in $\Phi$. This clearly holds for $i=0$.
Under the assumption that $\Win_i(s) \in \Phi$ for all $s$, we show that $\PreC_{\Win_i}(s)$ is equivalent to a formula in $\Phi$, from which the claim follows for the case of $i+1$.
Let $\phi :=\PreC_{\Win_i}(s)$ as defined in \eqref{eq:preCeq}, but where for distinction we rename quantified variables according to the number of $\forall$-quantifiers in whose scope they occur: let $v\in V$ be renamed to $v_i$ if it occurs below $i$ $\forall$-quantifiers.
Let $[\varphi]$ be obtained from $\varphi$ by eliminating all equality literals $x=y$ in $\varphi$ and uniformly substituting all variables in an equivalence class by some arbitrary representative.
By the induction hypothesis, the quantifier depth of $\Win(s_{ij})$ is bounded by $2K$ (counting both $\exists$ and $\forall$, which occur in an alternating fashion).
So the formula $[\varphi]$ encodes a tree of paths $\rho_1, \rho_2,\dots, \rho_l$ which represent variable dependencies that span at most $K+1$ instants, due to the two additional quantifiers added in $\PreC$. 
Since $\psi$ has $K$-bounded lookback, for each path $\rho_i$, there can be a dependency chain from the current variables $V$ to at most $V_K$. Thus, all atoms that mention $V_{K+1}$ must be irrelevant and can be removed, so that we obtain a formula equivalent to $[\phi]$ of quantifier depth $2K$.

Therefore all formulas $\Win_i(s)$ for and AND-node $s$ of $\ANDOR$ and $i\geq 0$ come from a set that is finite up to equivalence, so whenever computing a new formula $\Win_i(s)$, it can be checked whether it is equivalent to an already computed one.
Here equivalence of $\phi$ and $\phi'$ can be checked by verifying that $\phi \models_\TT \phi'$ and vice versa, i.e., that $\phi \wedge \neg \phi'$ is $\TT$-unsatisfiable.
This requires to check satisfiability of a FO formula with $2K$ quantifier alternations.
Thus the fixpoint $\Win(\ANDOR)$ must exist, and the synthesis problem can be solved via \thmref{correctness}.
\end{proof}

\end{document}